\documentclass[10pt]{article}
\usepackage[margin = 1in]{geometry}
\usepackage{amsmath}
\usepackage{amssymb}
\usepackage{multirow}
\usepackage{graphicx} 
\usepackage{epstopdf}

\usepackage{amsthm} 
\usepackage{physics,braket} 

\usepackage{hyperref} 

\usepackage[title]{appendix} 

\usepackage{tikz}
\usetikzlibrary{arrows,shapes,positioning}
\usetikzlibrary{decorations.markings}
\tikzstyle arrowstyle=[scale=1]
\tikzstyle directed=[postaction={decorate,decoration={markings,
    mark=at position .525 with {\arrow[arrowstyle]{stealth}}}}]

\newtheorem{theorem}{Theorem}
\newtheorem{proposition}{Proposition}

\DeclareMathOperator{\sgn}{sgn} 

\newcommand{\dual}[1]{#1^{\cross}} 
\newcommand{\operator}[1]{\mathsf{#1}}
\newcommand{\adjoint}[1]{\operator{#1}^\dagger}
\newcommand{\hilbertspace}{\mathcal{H}}
\newcommand{\rhsextension}[1]{\operator{#1}^{\cross}} 
\newcommand{\domain}[1]{\mathcal{D}(#1)}
\newcommand{\fourier}{\mathfrak{F}}

\newcommand{\reals}{\mathbb{R}}

\title{\bf Conjugates to One Particle Hamiltonians in 1-Dimension in Differential Form}
\author{Ralph Adrian E. Farrales\textsuperscript{a}, Herbert B. Domingo\textsuperscript{b}, and Eric A. Galapon\textsuperscript{a}\\ \\\textsuperscript{a}Theoretical Physics Group, National Institute of Physics\\University of the Philippines Diliman, Philippines\\ \\\textsuperscript{b}Laboratory for Applied Mathematical Physics\\Department of Physical Sciences and Mathematics\\University of the Philippines Manila, Philippines}

\begin{document}

\maketitle

\begin{abstract}
    A time operator is a Hermitian operator that is canonically conjugate to a given Hamiltonian. For a particle in 1-dimension, a Hamiltonian conjugate operator in position representation can be obtained by solving a hyperbolic second-order partial differential equation, known as the time kernel equation, with some boundary conditions. One possible solution is the time of arrival operator. Here, we are interested in finding other Hamiltonian conjugates by further studying the boundary conditions. A modified form of the time kernel equation is also considered which gives an even bigger solution space.
\end{abstract}

\section{Introduction}
If time was an observable, then, in standard quantum mechanics, it would be represented by an operator that is (i) self-adjoint or at least Hermitian; and is (ii) canonically conjugate to the Hamiltonian. The former assures us that the eigenvalues are real, and the latter is a consequence of its dynamics. Since our time operator $\operator{T}$ evolves as $\dv*{\operator{T}}{t} = \pm 1$ (where the sign depends on whether it increases or decreases in step with parametric time), then the Heisenberg equation of motion gives the time-energy canonical commutation relation
\begin{equation} \label{eq:teccr}
    [\operator{T},\operator{H}] = \pm i\hbar \, ,
\end{equation}
where $\operator{H}$ is the Hamiltonian. In contrast with parametric time, this time operator $\operator{T}$ contains a dynamical aspect and usually encompasses questions regarding the duration of an event, or the time of event occurrence, whose value changes as parametric time changes. The relation \eqref{eq:teccr} is a consequence of Dirac's correspondence principle between the Poisson bracket and the commutator \cite{gotay2000}. This commutation relation between time and energy is also connected to the time-energy uncertainty relation $\Delta\operator{T} \Delta\operator{H} \geq \hbar/2$, and contributes to some of the many possible interpretations on $\Delta \operator{T}$ and $\Delta \operator{H}$ \cite{muga2008}.

Finding such time operator has been met with several obstacles; in general, the problem of time in quantum mechanics has been a subject of much controversy throughout the years \cite{muga2008,muga2009}. Very early on, Pauli's infamous argument denying the existence of such operator \cite{pauli1980} shaped much of the research, and pushed towards a nonstandard approach in dealing with time. In his footnote regarding the Heisenberg equation of motion, Pauli rejected the existence of a Hermitian operator $\operator{T}$ which satisfies 
\eqref{eq:teccr} (the unitarity of $\exp(-iE\operator{T}/\hbar)$ indicates that he actually meant that there exists no \textit{self-adjoint} operator $\operator{T}$) relegating time as merely an ordinary number. However, it was rigorously shown that there is no inconsistency in assuming a bounded self-adjoint time operator conjugate to the Hamiltonian with a semi-bounded, unbounded, or finitely countable spectrum \cite{galapon2002}. We note though that while the self-adjointness of $\operator{T}$ is also a desired property, it will be sufficient to find Hermitian operators $\operator{T}$ (without any detailed analysis of the domains) satisfying the commutation relation \eqref{eq:teccr}. 

Various authors have worked on finding such operator $\operator{T}$: using quantization \cite{aharonov1961,galapon2018}, using a wave function in time-representation \cite{kijowski1974}, the energy shift operator \cite{bauer1983}, and a partial derivative with respect to the energy \cite{razavy1969,olkhovsky1974,goto1981}. Here, we highlight three methods of finding operators conjugate to the Hamiltonian: the work of Bender and Dunne using basis operators \cite{bender1989exact,bender1989integration}, the work of Galapon using supraquantization \cite{galapon2004}, and the work of Galapon and Villanueva using the Liuoville superoperator \cite{galapon2008}.

Bender and Dunne's solution \cite{bender1989exact,bender1989integration} constitutes expanding both $\operator{T}$ and $\operator{H}$ in terms of basis operators
\begin{align} 
    \operator{T}_{m,n} &= \frac{1}{2^n} \sum_{k=0}^n \binom{n}{k} \operator{q}^k \operator{p}^m \operator{q}^{n-k} \label{eq:bdbo1} \\
    &= \frac{1}{2^m} \sum_{j=0}^m \binom{m}{j} \operator{p}^j \operator{q}^n \operator{p}^{m-j} \, , \label{eq:bdbo2}
\end{align}
for $m \geq 0$ and $n \geq 0$. Note that the two forms are equivalent, due to the commutation relation between $\operator{q}$ and $\operator{p}$, i.e., $[\operator{q},\operator{p}]=i\hbar$. This can also be extended when either $m$ or $n$ is negative: for $n \geq 0$ and $m < 0$, we use \eqref{eq:bdbo1}; while for $m \geq 0$ and $n < 0$, we use \eqref{eq:bdbo2}. The Bender-Dunne basis operators \eqref{eq:bdbo1} and \eqref{eq:bdbo2} are the Weyl-ordered quantization of $p^m q^n$ and are densely defined operators in the Hilbert space $L^2(\reals)$ \cite{bunao2014}. Other basis operators such as the simple symmetric and Born-Jordan ordering are also possible choices \cite{domingo2015,bagunu2021}.

For a given Hamiltonian and for a given choice of basis operators, the coefficients of the expansion of $\operator{H}$ will already be known. The unknown operator $\operator{T}$ may then take the form
\begin{equation} \label{eq:benderdunne}
    \operator{T} = \pm \sum_{m,n} \alpha_{m,n} \operator{T}_{m,n} \, ,
\end{equation}
where the $\alpha_{m,n}$'s are solved by imposing \eqref{eq:teccr} (and Hermiticity if desired) and getting a recurrence relation. Note the non-uniqueness of this solution is due to the fact that we can add any operator $\operator{C}$ which commutes with the Hamiltonian to get another solution $\operator{T} + \operator{C}$. Bender and Dunne's minimal solution is obtained by vanishing as many $\alpha_{m,n}$'s as possible while still satisfying the recurrence relation generated by \eqref{eq:teccr}. Using the quantized $p^m q^n$ as basis operators is in contrast with doing quantization of the classical observable itself, wherein the latter does not guarantee that \eqref{eq:teccr} is satisfied \cite{gotay2000,galapon2018}.

Galapon's \cite{galapon2004} supraquantization approach provides a solution of \eqref{eq:teccr} in coordinate representation. While addressing the problems of quantization and of the quantum time of arrival and inspired by earlier efforts by Mackey \cite{mackey1968}, Galapon introduced the idea of supraquantization which aims to construct quantum observables using the axioms of quantum mechanics and the properties of the system. In supraquantization, the classical observable is the boundary condition, which is in contrast with quantization wherein the classical observable is treated as the starting point.

The supraquantized operators are constructed under the rigged Hilbert space $\dual{\Phi} \supset \hilbertspace \supset \Phi$ \cite{bohm1978}, and the generalized observable in $\Phi$-representation takes the integral form
\begin{equation}
	(\mathcal{T}\varphi)(q) = \int_{-\infty}^\infty \braket{ q | \mathcal{T} | q' } \varphi(q') \dd{q'} \, .
\end{equation}
The kernel is assumed to take the form
\begin{equation} \label{eq:tkintro}
	\braket{ q | \mathcal{T} | q' } = \frac{\mu}{i\hbar} T(q,q') \sgn(q - q') \, ,
\end{equation}
inspired by the 1-dimensional time of arrival of a particle of mass $\mu$ under some continuous potential $V(q)$. The rigged extension of \eqref{eq:teccr} imposes that for $\mathcal{T}$ to be conjugate to the rigged Hilbert space extension of $\operator{H}$, then the kernel factor $T(q,q')$ should satisfy
\begin{equation} \label{eq:tkeintro}
    -\frac{\hbar^2}{2\mu} \, \pdv[2]{T(q,q')}{q} + \frac{\hbar^2}{2\mu} \, \pdv[2]{T(q,q')}{{q'}} + [V(q) - V(q')] T(q,q') = 0 \, ,
\end{equation}
and
\begin{equation} \label{eq:tkegbcpm}
	\dv{T(q,q)}{q} + \pdv{T(q,q')}{q} \bigg|_{q'=q} + \pdv{T(q,q')}{{q'}} \bigg|_{q'=q} = \mp 1 \, ,
\end{equation}
where, in the right hand side of \eqref{eq:tkegbcpm}, $-1$ is used when increasing with step in time, and $+1$ when decreasing (the latter being used when getting the time of arrival solution). The hyperbolic second-order partial differential equation \eqref{eq:tkeintro} is referred to as the time kernel equation \cite{galapon2004}, and, together with condition \eqref{eq:tkegbcpm}, defines a family of solutions which are canonically conjugate to the Hamiltonian. The quantum time of arrival is but one possible solution of \eqref{eq:tkeintro} and \eqref{eq:tkegbcpm}; and hence, just one example of a Hermitian operator conjugate to the Hamiltonian.

Galapon and Villanueva's Liouville solution \cite{galapon2008} uses the Liouville superoperator $\mathcal{L}_\operator{A} = [\operator{A}, \cdot]$ so that the solution takes the form
\begin{equation} \label{eq:liouville}
    \operator{T} = \mp \mathcal{L}_\operator{H}^{-1} (i\hbar\operator{1}) \, .
\end{equation}
Dividing $\mathcal{L}_\operator{H}$ into its kinetic and potential parts, $\mathcal{L}_\operator{K}$ and $\mathcal{L}_\operator{V}$, enables a geometric expansion of \eqref{eq:liouville}. The domain is to be restricted to the Bender-Dunne operators \eqref{eq:benderdunne} so that the inverse can be well-defined. The Liouville solution is then given by
\begin{equation}
    \operator{T} = \pm \mu \sum_{k=0}^\infty (-1)^k \qty(\mathcal{L}_\operator{K}^{-1} \mathcal{L}_\operator{V})^k \operator{T}_{-1,1} \, .
\end{equation}
It turns out that this is just the quantum version of the classical time of arrival at the origin, and is in the same vein as supraquantization. It was observed that for linear systems, the Liouville solution and the Bender-Dunne minimal solution coincided for $\alpha_{-1,1} = \mu$. The Liouville solution is equal to the Weyl quantization of the local time of arrival (expansion of the time of arrival about the free solution) for linear systems; meanwhile, only the leading term is equal to Weyl quantization for nonlinear systems. In coordinate representation, the Liouville kernel satisfies both \eqref{eq:tkeintro} and \eqref{eq:tkegbcpm} for everywhere analytic potentials. This Liouville solution highlights the connection between the Bender-Dunne solution and the supraquantized solution via the time kernel equation.

While much of the focus has been on the time of arrival solution, it should be noted that the time of arrival is not the only possible time observable that exists. The multi-faceted nature of time means that there are other time observables that will satisfy the commutation relation \eqref{eq:teccr}, not just the time of arrival. To reiterate, if $\operator{T}$ satisfies \eqref{eq:teccr}, then one can construct another operator, $\operator{T} + \operator{C}$, that is also a solution to \eqref{eq:teccr}, where $\operator{C}$ commutes with the Hamiltonian. One example is by choosing $\operator{C} = f(\operator{H})$ where $f$ is some suitable function of the Hamiltonian $\operator{H}$. Our interest is then to investigate these other conjugate solutions.

In this paper, we focus our attention on constructing a conjugate operator in $\Phi$-representation under the supraquantization approach. That is, we shall be using the differential equation \eqref{eq:tkeintro} and the condition \eqref{eq:tkegbcpm} to study these other Hamiltonian conjugate solutions. Solutions to this time kernel equation \eqref{eq:tkeintro}, satisfying the conjugate condition \eqref{eq:tkegbcpm}, and satisfying the additional boundary conditions $T(q,q) = q/2$ and $T(q,-q) = 0$ along the diagonal, corresponds to the quantum time of arrival at the origin of a particle in 1-dimension under potential $V(q)$. The solution is unique, is conjugate to the Hamiltonian, is Hermitian, satisfies time-reversal symmetry, and reduces to the classical time of arrival in the classical limit \cite{galapon2004}. We would like also to further study the condition \eqref{eq:tkegbcpm} to extract a more general form of boundary conditions for $T(q,q)$ and $T(q,-q)$, so that the other conjugate solutions can be constructed using methods in \cite{galapon2004}. Since the method is intertwined with the time of arrival, we restrict ourselves to operators which decrease in step with time, i.e., for $-1$ in the right hand side of \eqref{eq:teccr}, and $+1$ in the right hand side of \eqref{eq:tkegbcpm}. The negative of our solution will simply be the corresponding operator which increases in step with time (e.g., time of flight).

The rest of the paper is organized as follows. In Section \ref{sec:tke}, we briefly review the time kernel equation of \cite{galapon2004}. In Section \ref{sec:conjugate}, we derive the general form of $T(q,q)$ and $T(q,-q)$ which satisfy \eqref{eq:tkegbcpm}, prove the existence and uniqueness of the solution, and get conditions for Hermiticity and time-reversal symmetry. In Section \ref{sec:examples}, we look at some interesting examples of other conjugate solutions, primarily the time of arrival plus negative powers of the Hamiltonian. In Section \ref{sec:mtke}, we explore a modified form of the time kernel equation \eqref{eq:tkeintro} and condition \eqref{eq:tkegbcpm} that was introduced by Domingo \cite{domingo2004}, which removes the assumption of the form of the kernel \eqref{eq:tkintro}. We look at the modified conjugate solutions generated here, and compare it with the original time kernel equation solution. In Section \ref{sec:conclusion}, we conclude.

\section{The Time of Arrival Operator in 1-Dimension and the Time Kernel Equation} \label{sec:tke}

Classically, the time for a particle initially at $(q,p)$ in phase space at time $t = 0$ to arrive at some other point $x$ in the configuration space is given by
\begin{equation} \label{eq:ctoa}
	T_x(q,p) = -\sgn(p) \sqrt{\frac{\mu}{2}} \int_x^q \frac{\dd{q'}}{\sqrt{H(q,p) - V(q')}} \, ,
\end{equation}
where $\sgn$ is the signum function, $\mu$ is the mass of the particle, $H(q,p)$ is the Hamiltonian, and $V(q)$ is the interaction potential. We are only interested in the region $\Omega = \Omega_q \times \Omega_p$ where \eqref{eq:ctoa} is real-valued. In this region, we can expand $T_x$ about the free particle solution and get the local time of arrival $t_x(q,p)$ \cite{galapon2004}
\begin{equation} \label{eq:ltoa}
	t_x(q,p) = \sum_{k=0}^\infty (-1)^k T_k(q,p;x) \, ,
\end{equation}
where $T_k(q,p;x)$ satisfies the recurrence relation
\begin{equation} \label{eq:ltoarecur}
\begin{split}
	T_0(q,p;x) &= \frac{\mu}{p} (x - q) \, , \\
	T_k(q,p;x) &= \frac{\mu}{p} \int_q^x \dv{V}{q'} \, \pdv{T_{k-1}(q',p;x)}{p} \dd{q'} \, .
\end{split}
\end{equation}
For $p \neq 0$ and $V$ continuous at $q$, it was shown in \cite{galapon2004} that there exists a neighborhood of $q$ determined by $\abs{V(q) - V(q')} < K_\epsilon \leq p^2(2\mu)^{-1}$ such that for every $x$ in the said neighborhood of $q$, $t_x(q,p)$ converges absolutely and uniformly to $T_x(q,p)$. Since $T_x(q,p)$ holds in the entire $\Omega$, and $t_x(q,p)$ holds only in some local neighborhood $\omega_q$ of $\Omega_q$, then $t_x(q,p) \subset T_x(q,p)$, i.e., $T_x(q,p)$ is the analytic continuation of $t_x(q,p)$ in $\Omega \backslash \omega$. It is this local form of the time of arrival \eqref{eq:ltoa} that was supraquantized in \cite{galapon2004}.

The problem of supraquantizing the time of arrival conjugate with the Hamiltonian is treated in the rigged Hilbert space $\dual{\Phi} \supset \hilbertspace \supset \Phi$, where $\Phi$ is chosen in such a way that $\Phi$ is a dense subset of the domain of the Hamiltonian $\operator{H}$, and is invariant under $\operator{H}$. As with \cite{galapon2004}, we choose $\Phi$ to be the infinitely differentiable complex valued functions with compact support in the real line; $\dual{\Phi}$ is its corresponding dual space. The supraquantized time operator $\mathcal{T}: \Phi \to \dual{\Phi}$ is given by
\begin{equation} \label{eq:tintop}
	(\mathcal{T}\varphi)(q) = \int_{-\infty}^\infty \braket{ q | \mathcal{T} | q' } \varphi(q') \dd{q'} \, .
\end{equation}
Next, the transfer principle is hypothesized, which states that a particular property of one element of a class of observables can be transferred to the rest of the class \cite{galapon2004}. From the free particle solution, the kernel for all continuous potentials is then assumed to take the form
\begin{equation} \label{eq:tk}
	\braket{ q | \mathcal{T} | q' } = \frac{\mu}{i\hbar} T(q,q') \sgn(q - q') \, ,
\end{equation}
where $T(q,q')$ is real valued, symmetric $T(q,q') = T(q',q)$, and analytic.

Let $\rhsextension{H}$ be the extension of $\operator{H}$ in the entire $\dual{\Phi}$, i.e., $\rhsextension{H}: \dual{\Phi} \to \dual{\Phi}$ such that $\braket{\rhsextension{H}\phi|\varphi} = \braket{\phi|\adjoint{H}\varphi}$ for all $\phi$ in $\dual{\Phi}$ and $\varphi$ in $\Phi$, where $\adjoint{H}$ is the adjoint of $\operator{H}$ in $\Phi$. For $\dv*{\mathcal{T}}{t} = -\mathcal{I}$, the canonical commutation relation \eqref{eq:teccr} can be written as
\begin{equation} \label{eq:htccr}
	\braket{ \tilde{\varphi} | [ \rhsextension{H}, \mathcal{T} ] \varphi } = i\hbar \braket{ \tilde{\varphi} | \varphi } \, ,
\end{equation}
for all $\tilde{\varphi}$ and $\varphi$ in $\Phi$. The left hand side gives
\begin{equation} \label{eq:teccrtkegbc}
\begin{split}
	&\braket{ \tilde{\varphi} | [ \rhsextension{H}, \mathcal{T} ] \varphi } = i\hbar \int \tilde{\varphi}^*(q) \qty(\dv{T(q,q)}{q} + \pdv{T(q,q')}{q} \bigg|_{q'=q} + \pdv{T(q,q')}{{q'}} \bigg|_{q'=q}) \varphi(q) \dd{q} \\
	&\qquad + \frac{\mu}{i\hbar} \iint \tilde{\varphi}^*(q) \qty[ -\frac{\hbar^2}{2\mu} \, \pdv[2]{T(q,q')}{q} + V(q)T(q,q') ] \sgn(q-q') \varphi(q') \dd{q'}\dd{q} \\
	&\qquad - \frac{\mu}{i\hbar} \iint \tilde{\varphi}^*(q) \qty[ -\frac{\hbar^2}{2\mu} \, \pdv[2]{T(q,q')}{{q'}} + V(q')T(q,q') ] \sgn(q-q') \varphi(q') \dd{q'}\dd{q} \, .
\end{split}
\end{equation}
The integration is across the common support of $\phi(q)$ and $\tilde{\phi}(q)$. We then need $\Phi$ to be invariant under $\rhsextension{H}$; if $\operator{H}$ is self-adjoint, then this is true for $\Phi$ invariant under $\operator{H}$. The $\mathcal{T}\rhsextension{H}$ term, the last term in \eqref{eq:teccrtkegbc}, is obtained by doing integration by parts and noting that $\phi(q)$ and $\tilde{\phi}(q)$ vanish at the boundary of their respective compact supports.

For this to satisfy the right hand side of \eqref{eq:htccr}, the kernel factor $T(q,q')$ must satisfy the so-called time kernel equation \cite{galapon2004}
\begin{equation} \label{eq:tke}
    -\frac{\hbar^2}{2\mu} \, \pdv[2]{T(q,q')}{q} + \frac{\hbar^2}{2\mu} \, \pdv[2]{T(q,q')}{{q'}} + [V(q) - V(q')] T(q,q') = 0 \, ,
\end{equation}
and the boundary condition
\begin{equation} \label{eq:tkegbc}
	\dv{T(q,q)}{q} + \pdv{T(q,q')}{q} \bigg|_{q'=q} + \pdv{T(q,q')}{{q'}} \bigg|_{q'=q} = 1 \, .
\end{equation}
The boundary conditions along the diagonal of $T(q,q')$ can be fixed in such a way that it satisfies certain properties. For example, if one requires the kernel solution to satisfy conjugacy \eqref{eq:tkegbc}, to be Hermitian, to satisfy time-reversal symmetry, and to approach the local time of arrival at the origin via the Wigner-Weyl transform,
\begin{equation} \label{eq:wignerweyl}
	\mathcal{T}_\hbar(q,p) = \int_{-\infty}^\infty \Braket{ q + \frac{v}{2} | \mathcal{T} | q - \frac{v}{2} } \exp\left(-i\frac{vp}{\hbar}\right) \dd{v} \, , 
\end{equation}
in the classical limit of vanishing $\hbar$, then the specific boundary conditions should be \cite{galapon2004}
\begin{equation} \label{eq:tkebc}
	T(q,q) = \frac{q}{2} \, , \qquad T(q,-q) = 0 \, .
\end{equation}
Solutions of the time kernel equation \eqref{eq:tke} satisfying these boundary conditions \eqref{eq:tkebc} are unique for continuous potentials, and correspond to the quantum time of arrival operator at the origin (if the arrival point is elsewhere, we just shift the potential to move the arrival point to the origin). To illustrate, the following are some time of arrival kernel solutions for linear systems (i.e., linear equations of motion):

\paragraph{Free Particle} For $V(q) = 0$, the solution satisfying \eqref{eq:tke} and \eqref{eq:tkebc} is
\begin{equation} \label{eq:freetkesoln}
    T(q,q') = \frac{1}{4} (q + q') \, .
\end{equation}

\paragraph{Harmonic Oscillator} For $V(q) = \mu\omega^2 q^2/2$,
\begin{equation} \label{eq:hoscitkesoln}
    T(q,q') = \frac{1}{4} \sum_{j=0}^\infty \frac{1}{(2j+1)!} \qty(\frac{\mu \omega}{2\hbar})^{2j} (q+q')^{2j+1} (q-q')^{2j} = \frac{1}{4} \qty(\frac{2\hbar}{\mu\omega}) \frac{1}{q-q'} \sinh\qty(\frac{\mu\omega}{2\hbar}(q+q')(q-q')) \, .
\end{equation}

These solutions give the local time of arrival at the origin upon using the Wigner-Weyl transform \eqref{eq:wignerweyl} and the relation
\begin{equation} \label{eq:fouriertnsgnt}
    \int_{-\infty}^\infty t^n \sgn(t) e^{-it\omega} \dd{t} = \frac{2 n!}{(i\omega)^{n+1}} \, , \qquad \omega \neq 0 \, .
\end{equation}
It was also shown that for nonlinear systems, the solution approaches the local time of arrival for vanishing $\order{\hbar^2}$.

\section{The Boundary Conditions of the Time Kernel Equation} \label{sec:conjugate}
The time of arrival boundary conditions \eqref{eq:tkebc} just give one type of kernel solution which satisfies the general boundary condition \eqref{eq:tkegbc}---it is but one member of the family of conjugate solutions. This was due to the condition that the solution approaches the classical time of arrival in the classical limit. Other Hamiltonian conjugates may correspond to a different kind of boundary conditions along the diagonal.

It will be convenient to rewrite the time kernel equation in canonical form,
\begin{equation} \label{eq:tkecanon}
    -\frac{2 \hbar^2}{\mu} \pdv{T(u,v)}{u}{v} + \qty[ V\qty(\frac{u+v}{2}) - V\qty(\frac{u-v}{2}) ] T(u,v) = 0 \, ,
\end{equation}
where we have changed variables to $u = q + q'$ and $v = q - q'$. Throughout the paper, we will often be using this canonical form to solve for $T(u,v)$ via Frobenius method, where $T(q,q')$ can be retrieved by changing the variables back to $q$ and $q'$.

We will be looking at analytic solutions to \eqref{eq:tkecanon} of the form
\begin{equation} \label{eq:tkesolncanon}
    T(u,v) = \sum_{m,n} \alpha_{m,n} u^m v^n \, ,
\end{equation}
or equivalently, in $qq'$-coordinates,
\begin{equation} \label{eq:tkesoln}
    T(q,q') = \sum_{m,n} \alpha_{m,n} (q+q')^m (q-q')^n \, ,
\end{equation}
for nonnegative $m$ and $n$.

\subsection{Boundary Conditions for the Hamiltonian Conjugate Solutions}

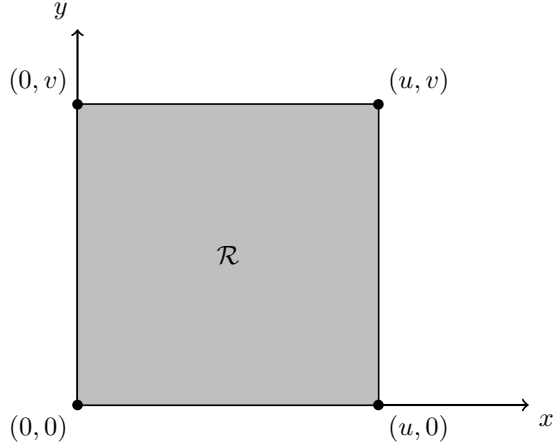
\begin{figure}
\centering

\begin{tikzpicture}

\draw[thick,->] (0,0) -- (6,0) node[anchor=north west] {$x$};
\draw[thick,->] (0,0) -- (0,5) node[anchor=south east] {$y$};

\draw[thick] (0,0) -- (4,0) node[anchor=north west] {$(u,0)$};
\draw[thick] (4,0) -- (4,4) node[anchor=south west] {$(u,v)$};
\draw[thick] (4,4) -- (0,4) node[anchor=south east] {$(0,v)$};
\draw[thick] (0,4) -- (0,0) node[anchor=north east] {$(0,0)$};

\draw[fill=gray!50] (0,0)--(4,0)--(4,4)--(0,4)--cycle;

\node at (2,2) {$\mathcal{R}$};

\fill (0,0) circle (2pt);
\fill (4,0) circle (2pt);
\fill (4,4) circle (2pt);
\fill (0,4) circle (2pt);

\end{tikzpicture}

\caption{Rectangular region $\mathcal{R}$ in the plane defined by $0 \leq x \leq u$ and $0 \leq y \leq v$.} \label{fig:region}
\end{figure}

It would be preferable to rewrite the general condition \eqref{eq:tkegbc} as boundary conditions along the diagonal $q' = q$ and $q' = -q$. This would let us use the methods found in \cite{galapon2004} to construct the other conjugate solutions. We then start by showing the required form of the boundary conditions along the diagonal that is essentially identical to \eqref{eq:tkegbc}. Since we will only be looking at entire analytic potentials, then it would be reasonable to restrict ourselves to analytic solutions of the time kernel equation. In canonical form \eqref{eq:tkecanon} and \eqref{eq:tkesolncanon}, the boundary conditions will be along the axis $u = 0$ and $v = 0$.

\begin{theorem}
Analytic solutions \eqref{eq:tkesolncanon} to the time kernel equation \eqref{eq:tkecanon} satisfy the general boundary condition \eqref{eq:tkegbc} if and only if they satisfy the specific boundary conditions
\begin{equation} \label{eq:tkebchccanon}
    T(u,0) = \frac{u}{4} + c \, , \qquad T(0,v) = g(v) + c \, ,
\end{equation}
where $c$ is a constant, $g$ is some differentiable function, and $g(0)=0$.
\end{theorem}

\begin{proof}
Going back to $(q,q')$, we substitute the assumed form of the solution \eqref{eq:tkesoln} to the general boundary condition \eqref{eq:tkegbc} to get
\begin{equation}
    \sum_m \alpha_{m,0} 2^{m+1} m q^{m-1} = 1 \, .
\end{equation}
To satisfy this, we need to set $\alpha_{1,0} = 1/4$, $\alpha_{2,0} = \alpha_{3,0} = \cdots = 0$. We can set the $\alpha_{0,n}$'s to be any constant for $n\geq0$. Thus, for $\alpha_{0,0} = c$, $\alpha_{0,k} = \beta_k$, and $g(v) = \sum_{k=1}^\infty \beta_k v^k$, we see that \eqref{eq:tkesoln} gives $T(q,q) = \alpha_{0,0} + \alpha_{1,0} (2q) = c + q/2$ and $T(q,-q) = \sum_n \alpha_{0,n} (2q)^n = c + g(2q)$. In canonical form, we get $T(u,0) = u/4 + c$ and $T(0,v) = g(v) + c$, verifying \eqref{eq:tkebchccanon}.

Conversely, if $T(u,v)$ satisfies \eqref{eq:tkebchccanon}, then from \eqref{eq:tkecanon}, the integral form of the time kernel equation gives
\begin{equation} \label{eq:tkeinthccanon}
    T(u,v) = \frac{u}{4} + g(v) + c +  \frac{\mu}{2\hbar^2} \int_0^v \int_0^u \qty[V\qty(\frac{x+y}{2}) - V\qty(\frac{x-y}{2}) ] T(x,y) \dd{x} \dd{y} \, , 
\end{equation}
and, in $(q,q')$,
\begin{equation} \label{eq:tkeinthc}
\begin{split}
    T(q,q') &= \frac{1}{4}(q+q') + \sum_{k=1}^\infty \beta_k (q-q')^k + c \\
    &\quad + \frac{\mu}{2\hbar^2} \int_0^{q-q'} \int_0^{q+q'} \qty[V\qty(\frac{x+y}{2}) - V\qty(\frac{x-y}{2}) ] T(x,y) \dd{x} \dd{y} \, , 
\end{split}
\end{equation}
where, again, we let $g(v) = \sum_{k=1}^\infty \beta_k v^k$. The integration is evaluated over the region in the plane defined by $0 \leq x \leq u$ and $0 \leq y \leq v$ (see Figure \ref{fig:region}). Then, using the Leibniz integral rule, we get $\dv{q}T(q,q) = 1/2$, $\pdv{q}T(q,q')|_{q'=q} = 1/4 + \beta_1$, and $\pdv{q'}T(q,q')|_{q'=q} = 1/4 - \beta_1$, i.e., $T(q,q')$, and hence $T(u,v)$, satisfies the general boundary condition \eqref{eq:tkegbc}.
\end{proof}

In $qq'$-coordinates, the boundary conditions \eqref{eq:tkebchccanon} can be written as
\begin{equation} \label{eq:tkebchc}
    T(q,q) = \frac{q}{2} + c \qquad\text{and}\qquad T(q,-q) = g(2q) + c \, ,
\end{equation}
where $c$ is a constant, $g$ is some differentiable function, and $g(0)=0$. Note that we have the awkward notation $g(2q)$, since, from the canonical form, we expect it to be $g(v) \to g(q-q')$, and the $T(q,-q)$ condition makes this $g(2q)$.

Therefore, solutions to the time kernel equation \eqref{eq:tke} satisfying the boundary conditions \eqref{eq:tkebchc} also satisfy the general boundary condition \eqref{eq:tkegbc}, and thus corresponds to operators canonically conjugate to the Hamiltonian. Conversely, solutions satisfying the general condition \eqref{eq:tkegbc} also satisfy \eqref{eq:tkebchc}. This means that the boundary condition \eqref{eq:tkebchc} constitute the \textit{same} family of Hamiltonian conjugates as that of \eqref{eq:tkegbc}. One could then use Frobenius method to construct solutions of the time kernel equation \eqref{eq:tke} satisfying \eqref{eq:tkebchc}. The Wigner-Weyl transform \eqref{eq:wignerweyl} will be used to see its classical limit.

\subsection{Existence and Uniqueness}

Note that the rest of the $\alpha_{m,n}$'s are obtained from a recurrence relation imposed by the time kernel equation. The above theorem assumed that this analytic solution exists. We shall focus our attention only on continuous potentials. We prove below the existence and uniqueness of the solution for these types of potentials \cite{freiling2008} (using analogous steps in \cite{domingo2004}).

\begin{theorem} \label{thm:tkeunique}
If $V(q)$ is a continuous function at any point of the real line, then there exists a unique continuous solution $T(q,q')$ to the time kernel equation \eqref{eq:tke} satisfying the boundary conditions \eqref{eq:tkebchc}.
\end{theorem}

\begin{proof}
In canonical form, we have the time kernel equation \eqref{eq:tkecanon} and the boundary conditions \eqref{eq:tkebchccanon}, with the integral form of the time kernel equation given by \eqref{eq:tkeinthccanon}.

Let $T(u,v) = \lim_{n \to \infty} T_n(u,v)$, where
\begin{equation}
    T_0(u,v) = \frac{u}{4} + g(v) + c \, ,
\end{equation}
and
\begin{equation}
    T_n(u,v) = \frac{u}{4} + g(v) + c + \frac{\mu}{2\hbar^2} \int_0^v \int_0^u \qty[V\qty(\frac{x+y}{2}) - V\qty(\frac{x-y}{2}) ] T_{n-1}(x,y) \dd{x} \dd{y}  \, .
\end{equation}
We can also write $T_n$ as
\begin{equation}
    T_n(u,v) = \frac{u}{4} + g(v) + c + \sum_{j=1}^n \qty[T_j(u,v) - T_{j-1}(u,v)] \, .
\end{equation}
Since $V(q)$ is continuous, there exists an $M > 0$ such that $\abs{V\qty(\frac{x+y}{2}) - V\qty(\frac{x-y}{2})} \leq M$ for any point $(x,y)$ in the region $\mathcal{R}$ in Figure \ref{fig:region}. Since $g$ is differentiable, we could write it as $g(v) = \sum_{k=1}^\infty \beta_{k} v^k$. To simplify things, we let $g(v) + c = \sum_{k=0}^\infty \beta_{k} v^k$, where $\beta_0 = c$. We then have,
\begin{align}
    \abs{T_1(u,v) - T_0(u,v)} &\leq \frac{\mu M}{2 \hbar^2} \int_0^v \int_0^u \abs{\frac{x}{4} + \sum_k \beta_k y^k} \dd{x} \dd{y} \nonumber \\
    &\leq \frac{\mu M}{2 \hbar^2} \int_0^v \int_0^u \qty(\abs{x} + \sum_k \abs{\beta_k}\abs{y^k}) \dd{x} \dd{y} \nonumber \\
    &\leq \frac{\mu M}{2 \hbar^2} \qty(\abs{u}^2 \abs{v} + \sum_k \abs{\beta_k} \abs{u} \abs{v}^{k+1} \frac{1}{k+1}) \nonumber \\
    &\quad = \frac{\mu M}{2 \hbar^2} \abs{u}\abs{v} \qty(\abs{u} + \sum_k \abs{\beta_k} \frac{\abs{v}^k}{k+1}) \, .
\end{align}
Similarly,
\begin{align}
    \abs{T_2(u,v) - T_1(u,v)} &\leq \frac{\mu M}{2 \hbar^2} \int_0^v \int_0^u \abs{T_1(x,y) - T_0(x,y)} \dd{x} \dd{y} \nonumber \\
    &\leq \qty(\frac{\mu M}{2 \hbar^2})^2 \frac{\abs{u}^2 \abs{v}^2}{2} \qty(\abs{u} + \sum_k \abs{\beta_k} \frac{\abs{v}^k}{(k+1)(k+2)}) \, ,
\end{align}
and,
\begin{equation}
    \abs{T_3(u,v) - T_2(u,v)} \leq \qty(\frac{\mu M}{2 \hbar^2})^3 \frac{\abs{u}^3 \abs{v}^3}{2 \cdot 3} \qty(\abs{u} + \sum_k \abs{\beta_k} \frac{\abs{v}^k}{(k+1)(k+2)(k+3)}) \, .
\end{equation}
By induction,
\begin{equation}
    \abs{T_j - T_{j-1}} \leq \qty(\frac{\mu M}{2 \hbar^2})^j \frac{\abs{u}^j \abs{v}^j}{j!} \qty(\abs{u} + \sum_k \abs{\beta_k} \frac{\abs{v}^k}{(k+1)(k+2)\dotsm(k+j)} ) \, .
\end{equation}
We see that this goes to zero as $j$ approaches infinity. Therefore, $T_n(u,v)$ is absolutely and uniformly convergent for all finite values of $u$ and $v$, and the limit $T(u,v)$ defines the solution to the time kernel equation.

Now, suppose $T_a(u,v)$ and $T_b(u,v)$ are two solutions to the time kernel equation. The existence of a continuous solution implies that there exists a $K > 0$ such that $\abs{T_a(u,v)} \leq K$ and $\abs{T_b(u,v)} \leq K$ in the bounded region in Figure \ref{fig:region}. Then, from their integral forms \eqref{eq:tkeinthccanon}, we get as a first approximation,
\begin{align}
    \abs{T_a(u,v) - T_b(u,v)} &\leq \frac{\mu M}{2 \hbar^2} \int_0^v \int_0^u \abs{T_a(x,y) - T_b(x,y)} \dd{x} \dd{y} \nonumber \\
    &\leq \frac{\mu M}{2 \hbar^2} (2K) \abs{u}\abs{v} \, .
\end{align}
We substitute this back to get a second approximation,
\begin{equation}
    \abs{T_a(u,v) - T_b(u,v)} \leq \qty(\frac{\mu M}{2\hbar^2})^2 (2K) \frac{\abs{u}^2 \abs{v}^2}{2} \, .
\end{equation}
By induction, the $n$th approximation gives
\begin{equation}
    \abs{T_a(u,v) - T_b(u,v)} \leq \qty(\frac{\mu M}{2\hbar^2})^n (2K) \frac{\abs{u}^n \abs{v}^n}{n!} \, .
\end{equation}
This approaches zero as $n$ approaches infinity. Therefore, the solution is unique.
\end{proof}

\subsection{Hermiticity and Time-Reversal Symmetry} \label{subsec:hermtime}

Since we are interested in constructing Hermitian operators conjugate to the Hamiltonian, we need to impose additional conditions to the time kernel equation solutions to satisfy Hermiticity. To compare with the time of arrival solution, we will also be looking at the conditions for time-reversal symmetry.

\paragraph{Hermiticity} It would be sufficient to construct conjugate operators that are Hermitian. The supraquantized time operator is Hermitian if $\adjoint{\mathcal{T}} = \mathcal{T}$. From \eqref{eq:tintop} and \eqref{eq:tk} in canonical form, the Hermiticity condition can be rewritten as
\begin{equation}
    T(u,v) = T^*(u,-v) \, .
\end{equation}
We look at the integral form of $T(u,v)$ in \eqref{eq:tkeinthccanon}, and get
\begin{equation}
    T^*(u,-v) = \frac{u}{4} + g^*(-v) + c^* +  \frac{\mu}{2\hbar^2} \int_0^v \int_0^u \qty[V\qty(\frac{x+y}{2}) - V\qty(\frac{x-y}{2}) ]^* T^*(x,-y) \dd{x} \dd{y} \, .
\end{equation}
So, if $T(u,v) = T^*(u,-v)$ (and taking $T(x,y) = T^*(x,-y)$ to be true inside the integral as well), then a sufficient condition will be obtained by equating each term, giving
\begin{equation}
    g(v) = g^*(-v) \, ,
\end{equation}
\begin{equation} \label{eq:chermitian}
    c = c^* \, ,
\end{equation}
\begin{equation} \label{eq:vhermitian}
    V(q) = V^*(q) \, .
\end{equation}
With $g(v) = \sum_k \beta_k v^k$, we can rewrite the first condition as
\begin{equation} \label{eq:bhermitian}
    \beta_k = \beta_k^* (-1)^k \, .
\end{equation}
Thus, $c$ and $V(q)$ should be purely real, and the coefficients of the expansion of $g(v) = \sum_k \beta_k v^k$ should satisfy $\beta_k = \beta_k^* (-1)^k$. One way of satisfying the condition for $\beta_k$ is by letting $\beta_k = i^k \tilde{\beta}_k$, where $\tilde{\beta}_k$ is purely real. Another way is by letting $\beta_k$ vanish for odd $k$ and letting $\beta_k$ be purely real for even $k$.

\paragraph{Time-reversal symmetry} For completeness, we shall look at the conditions for time-reversal symmetry. The time operator satisfies time reversal symmetry if $\Theta\mathcal{T}\Theta^{-1} = -\mathcal{T}$ where $\Theta$ is the time-reversal operator, which in canonical form, gives
\begin{equation}
    T(u,v) = T^*(u,v) \, .
\end{equation}
Similar with the Hermiticity conditions, we can get the conditions for time-reversal symmetry to be
\begin{equation}
    g(v) = g^*(v) \, ,
\end{equation}
\begin{equation} \label{eq:ctimereversal}
    c = c^* \, ,
\end{equation}
\begin{equation} \label{eq:vtimereversal}
    V(q) = V^*(q) \, .
\end{equation}
Thus, we require that $g(v)$, $c$, and $V(q)$ be purely real. The first condition is equivalent to having its power series coefficients satisfy
\begin{equation} \label{eq:btimereversal}
    \beta_k = \beta_k^* \, ,
\end{equation}
for all $k$, that is, the $\beta_k$'s are all purely real.

\paragraph{Hermiticity and Time-reversal symmetry} The operator will satisfy both Hermiticity and time-reversal symmetry if the time kernel equation solution $T(u,v)$ satisfies all the above conditions. We see that $g$ should satisfy
\begin{equation}
    g(v) = g^*(v) = g^*(-v) \, ,
\end{equation}
meaning that it is purely real and even. Looking at its expansion coefficients, both $\beta_k$ conditions \eqref{eq:bhermitian} and \eqref{eq:btimereversal} are satisfied if
\begin{equation}
    (-1)^k = 1 \, ,
\end{equation}
which can only be satisfied for even $k$. Therefore, the solution is both Hermitian and satisfies time reversal symmetry if
\begin{equation}
    \beta_k = 0 \, , \qquad \text{for odd $k$,}
\end{equation}
\begin{equation}
    \beta_k = \beta_k^* \, , \qquad \text{for even $k$,}
\end{equation}
\begin{equation} 
    c = c^* \, ,
\end{equation}
\begin{equation} 
    V(q) = V^*(q) \, .
\end{equation}
where the $\beta_k$'s are from $g(v) = \sum_k \beta_k v^k$.

\paragraph{General Free Particle Solution} Let's take a look at the general solution \eqref{eq:tkeinthccanon} for the free particle case $V(q) = 0$,
\begin{equation}
    T(u,v) = \frac{u}{4} + c + g(v) \, ,
\end{equation}
which gives
\begin{equation}
    \braket{q|\mathcal{T}|q'} = \frac{\mu}{i\hbar}\sgn(q-q) \qty(\frac{q+q'}{4} + c + g(q-q')) \, .
\end{equation}
The Wigner-Weyl transform \eqref{eq:wignerweyl} gives
\begin{equation}
    \mathcal{T}_\hbar(q,p) = -\frac{\mu q}{p} - \frac{2\mu}{p}c + \frac{\mu}{i\hbar}\int_{-\infty}^\infty g(v) \sgn(v) \exp\qty(-iv\frac{p}{\hbar})\dd{v} \, ,
\end{equation}
where \eqref{eq:fouriertnsgnt} was used. Using $g(v) = \sum_{k=1}^\infty \beta_k v^k$, the above equation becomes
\begin{equation}
    \mathcal{T}_\hbar(q,p) = -\frac{\mu q}{p} - \frac{2\mu}{p}c - \frac{2\mu}{p} \sum_{k=1}^\infty \beta_k \frac{k!}{i^k} \qty(\frac{\hbar}{p})^k \, .
\end{equation}
For purely real $c$, the first two terms satisfy Hermiticity and time-reversal symmetry. The time-reversal symmetry condition of $\beta_k = \beta_k^*$ leaves the third term to have imaginary components. The Hermiticity condition removes the imaginary components, for example by having $\beta_k = i^k \tilde{\beta}_k$ for purely real $\tilde{\beta}_k$. Both Hermiticity and time-reversal symmetry conditions leave the third term to be purely real and only have terms of order $\order{\hbar^{2k}}$. Analysis of other potentials will become tedious since method of successive approximation will be required, but one can continue hypothesizing the transfer principle so that these properties carry over.

\section{Explicit Examples} \label{sec:examples}
One simple way of generating other Hamiltonian conjugates is by adding a term which commutes with the Hamiltonian, e.g., when $\mathcal{T}_0$ is canonically conjugate to the Hamiltonian, then so is $\mathcal{T}_0 + \mathcal{T}_C$, where $\mathcal{T}_C$ commutes with the Hamiltonian; one example is $\mathcal{T}_C = f(\operator{H})$ where $f$ is some suitable function of the Hamiltonian. Note that the Hamiltonians here are extended to the rigged Hilbert space, though starting here we drop the $\cross$ notation for convenience. Thus, we get other possible kernel solutions of the form
\begin{equation} \label{eq:tkc}
	T(q,q') = T_0(q,q') + T_C(q,q') \, ,
\end{equation}
where $T_0(q,q')$ corresponds to a Hamiltonian conjugate, and $T_C(q,q')$ corresponds to an operator that commutes with the Hamiltonian. From \eqref{eq:teccrtkegbc}, we must require the kernel factor $T_C(q,q')$ in \eqref{eq:tkc} to satisfy the time kernel equation
\begin{equation} \label{eq:tkec}
    -\frac{\hbar^2}{2\mu} \, \pdv[2]{T_C(q,q')}{q} + \frac{\hbar^2}{2\mu} \, \pdv[2]{T_C(q,q')}{{q'}} + [V(q) - V(q')] T_C(q,q') = 0 \, ,
\end{equation}
and the boundary condition
\begin{equation} \label{eq:tkemgbc}
	\dv{T_C(q,q)}{q} + \pdv{T_C(q,q')}{q} \bigg|_{q'=q} + \pdv{T_C(q,q')}{{q'}} \bigg|_{q'=q} = 0 \, ,
\end{equation}
where the right hand side is now zero (cf. \eqref{eq:tkegbc}). This $T_C(q,q')$ solution corresponds to operators $\mathcal{T}_C$ which commute with the Hamiltonian, i.e., $\mathcal{T}_C$ is a constant of motion. 

We now proceed with examples of other Hamiltonian conjugate solutions. We look at different boundary conditions along the diagonal of the form \eqref{eq:tkebchc} which satisfy the Hermiticity conditions \eqref{eq:bhermitian}, \eqref{eq:chermitian}, and \eqref{eq:vhermitian}. However, for all the examples except the last one, we do not impose time-reversal symmetry. In this paper, we are mostly interested in operators which are Hermitian and are conjugate to the Hamiltonian.

\subsection{Reciprocal of the Hamiltonian}
\paragraph{Free Particle} Consider, for $V(q)=0$, the solution
\begin{equation}
    T(q,q') = \frac{1}{4}(q+q') - i \beta (q-q') \, ,
\end{equation}
where $\beta$ is some constant. The first term is the usual free time of arrival solution. The second term is also a solution to the free time kernel equation which satisfies \eqref{eq:tkemgbc} instead. Note that $T(q,q') = T^*(q',q)$. This solution satisfies the boundary conditions
\begin{equation} \label{eq:tkehbc}
    T(q,q) = \frac{q}{2} \qquad \text{and} \qquad T(q,-q) = - 2 i \beta q \, .
\end{equation}
This is just \eqref{eq:tkebchc} with $c=0$ and $g(2q) = -i\beta (2q)$. This means that the solution is unique and corresponds to a Hamiltonian conjugate. From Section \ref{subsec:hermtime}, for purely real $\beta$, the solution is Hermitian, which one clearly sees from $T(q,q')$; however, the solution does not satisfy time-reversal symmetry since $T(q,-q)$ is purely imaginary.

The solution is in the form $T(q,q') = T_{\text{TOA}}(q,q') + T_C(q,q')$, where $T_{\text{TOA}}$ is the usual time of arrival solution, and $T_C$ is a solution satisfying \eqref{eq:tkemgbc}. Using \eqref{eq:tk}, we get the kernel
\begin{equation}
    \braket{ q | \mathcal{T} | q' } = \frac{\mu}{4i\hbar} (q + q') \sgn(q - q') - \frac{\mu \beta}{\hbar} (q - q') \sgn(q - q') \, ,
\end{equation}
which is in the form 
The Wigner-Weyl transform \eqref{eq:wignerweyl} of this kernel is given by
\begin{equation}
    \mathcal{T}_\hbar(q,p) = -\frac{\mu q}{p} + \beta \frac{2\mu\hbar}{p^2} \, ,
\end{equation}
for $p \neq 0$, where we have used the relation \eqref{eq:fouriertnsgnt}. We see that the first term is the free classical time of arrival (at the origin), while the second term $T_C$ becomes $\beta\hbar H^{-1}$, where $H(q,p) = p^2/2\mu$ is the free Hamiltonian. In the classical limit, this second term vanishes as $\hbar \to 0$ unless $\beta$ is sufficiently large, in the sense that $\beta \propto \hbar^{-1}$.

Note that $\hbar/H$ has the correct units of time. Therefore, the boundary conditions \eqref{eq:tkehbc} correspond to the time of arrival \textit{shifted} by $\beta\hbar H^{-1}$. One could interpret the shift term in the kernel solution to be the supraquantization of $\beta\hbar H^{-1}$, i.e., for free particle,
\begin{equation}
    \beta\hbar \braket{q|\mathcal{T}_C|q'} = - \frac{\mu \beta}{\hbar} (q - q') \sgn(q - q') \, .
\end{equation}
And so, when we set $\beta = 0$, we will obviously recover the time of arrival solution as \eqref{eq:tkehbc} becomes the time of arrival boundary conditions. When we apply the transfer principle \cite{galapon2004}, the boundary conditions \eqref{eq:tkehbc} must then correspond to a solution in the form of $T_{\text{TOA}} + \beta\hbar H^{-1}$ for all continuous potentials.

\paragraph{Harmonic Oscillator} The boundary conditions \eqref{eq:tkehbc} become
\begin{equation}
    T(u,0) = \frac{u}{4} \qquad \text{and} \qquad T(0,v) = -i\beta v \, ,
\end{equation}
in canonical form. We solve for $T(u,v)$ using Frobenius method, where the solution is of the form \eqref{eq:tkesolncanon}, and substituting to the time kernel equation \eqref{eq:tkecanon} will give a recurrence relation for $\alpha_{m,n}$ which depends on $V(q)$. The boundary conditions impose that
\begin{equation}
    \alpha_{m,0} = \frac{1}{4} \delta_{m,1} \qquad\text{and}\qquad \alpha_{0,n} = -i\beta \delta_{n,1} \, .
\end{equation}
We also let $\alpha_{m,n} = 0$ for $m < 0$ or $n < 0$ so that the solution is analytic.

For $V(q) = \frac{1}{2} \mu \omega^2 q^2$, the time kernel equation gives the recurrence relation
\begin{equation}
    \alpha_{m,n} = \frac{\mu^2 \omega^2}{4\hbar^2} \frac{1}{mn} \alpha_{m-2,n-2} \, .
\end{equation}
The condition $\alpha_{m,0} = \frac{1}{4} \delta_{m,1}$ produces nonvanishing values for $\alpha_{1,0}$, $\alpha_{3,2}$, $\alpha_{5,4}$, and so on---these terms are just that of the time of arrival solution $T_{\text{TOA}}$. The second condition $\alpha_{0,n} = -i\beta \delta_{n,1}$ produces a separate branch of nonvanishing terms
\begin{equation}
    \alpha_{0,1} = -i\beta \, ,
\end{equation}
\begin{equation}
    \alpha_{2,3} = -i\beta \frac{\mu^2 \omega^2}{4\hbar^2} \qty(\frac{1}{2}) \qty(\frac{1}{3}) \, ,
\end{equation}
\begin{equation}
    \alpha_{4,5} = -i\beta \qty(\frac{\mu^2 \omega^2}{4\hbar^2})^2 \qty(\frac{1}{2 \cdot 4}) \qty(\frac{1}{3 \cdot 5}) \, ,
\end{equation}
and so on. We can simplify this to a single index relation,
\begin{equation}
    \alpha_j = -i\beta \qty(\frac{\mu\omega}{2\hbar})^{2j} \frac{1}{(2j+1)!} \, ,
\end{equation}
and thus, the solution is then, upon changing back to $(q,q')$,
\begin{equation}
    T(q,q') = T_{\text{TOA}}(q,q') - i\beta \sum_{j=0}^\infty \frac{1}{(2j+1)!} \qty(\frac{\mu\omega}{2\hbar})^{2j} (q+q')^{2j} (q-q')^{2j+1} \, ,
\end{equation}
where the first term $T_{\text{TOA}}(q,q')$ is the usual harmonic oscillator time of arrival solution. One then gets the kernel \eqref{eq:tk} by multiplying $(\mu/i\hbar)\sgn(q-q')$. Upon using the Wigner-Weyl transform \eqref{eq:wignerweyl}, we get
\begin{equation}
    \mathcal{T}_\hbar(q,p) = t_{\text{TOA}}(q,p) + \beta \frac{2\mu\hbar}{p^2} \sum_{j=0}^\infty (-1)^j \qty(\mu^2 \omega^2 \frac{q^2}{p^2})^j \, , 
\end{equation}
where \eqref{eq:fouriertnsgnt} was used. We see again that the first term is the local time of arrival at the origin for the harmonic oscillator case. Meanwhile, the second term is a series expansion of $\beta\hbar H^{-1}$ for the harmonic oscillator Hamiltonian $H(q,p) = p^2/2\mu + \mu\omega^2 q^2 / 2$. This converges for $q$ close to the origin, in accordance with the local time of arrival (the shifting term converges for $\abs{q}<p/\mu\omega$). We therefore arrive at the shifted time of arrival as well for sufficiently large $\beta$.

The time kernel equation \eqref{eq:tke} and the boundary conditions \eqref{eq:tkehbc} then constitute the supraquantization of the series expansion of the classical time of arrival at the origin shifted by $\beta\hbar H^{-1}$. These supraquantized operators are Hermitian and are canonically conjugate to the their corresponding Hamiltonians. The shifting term takes the form
\begin{equation}
    \beta \hbar \braket{q|\mathcal{T}_C|q'} = \frac{\mu}{i\hbar} T_C(q,q') \sgn(q-q') \, .
\end{equation}
For the harmonic oscillator case, we see that
\begin{equation}
    \braket{q|\mathcal{T}_C|q'} = -\frac{\mu}{\hbar^2} \sum_{j=0}^\infty \frac{1}{(2j+1)!} \qty(\frac{\mu\omega}{2\hbar})^{2j} (q+q')^{2j} (q-q')^{2j+1} \sgn(q-q') \, .
\end{equation}

These solutions are a manifestation of the multiple solutions of $[\operator{H},\mathcal{T}] = \pm i\hbar$, i.e., we can add any multiple of $\hbar \mathcal{T}_C$ to $\mathcal{T}$ and the commutation relation will still be satisfied. The multiple solutions to the time-energy canonical commutation relation can be differentiated by looking at their dynamics and their internal symmetry \cite{caballar2009,caballar2010}. For example, one can see that if $\beta \neq 0$, then $\braket{q|\mathcal{T}|q'}^* \neq -\braket{q|\mathcal{T}|q'}$, i.e., time reversal symmetry is broken. This is in contrast with the time of arrival solution (the $\beta = 0$ case) where time-reversal symmetry is satisfied.

Finally, we note that when we set the first condition to $T(q,q) = 0$ (setting this gets rid of the time of arrival term), we have a method of constructing the supraquantization of just the reciprocal of the Hamiltonian. One can see that, for linear systems, its supraquantization $\braket{q|\mathcal{T}_C|q'}$ is equivalent to the corresponding Weyl quantization $\braket{q|\operator{H}^{-1}|q'}$ \cite{caballar2010,magadan2018}, consistent with the supraquantization of the local time of arrival being equal to its Weyl quantized form \cite{galapon2004}. However, it is only equal to the leading term for nonlinear systems.

\subsection{Negative powers of the Hamiltonian}
We can generalize the above example and consider the boundary conditions
\begin{equation} \label{eq:tkehnbc}
    T(q,q) = \frac{q}{2} \qquad \text{and} \qquad T(q,-q) = - i^{2N-1} \beta \mu^{-2(N-1)} (2q)^{2N-1} \, .
\end{equation}
for any integer $N \geq 1$. Since these are also a special case of \eqref{eq:tkebchc}, then the solutions also correspond to Hamiltonian conjugates. Since this satisfies \eqref{eq:bhermitian}, i.e., $(i^{2N-1}\tilde{\beta}_{2N-1})^* = (-1)^{2N-1} (i^{2N-1}\tilde{\beta}_{2N-1})$ for purely real $\tilde{\beta}_{2N-1} = \beta$, we also note that the solutions will be Hermitian. However, the solutions will not satisfy time-reversal symmetry.

Changing variables to $u=q+q'$ and $v=q-q'$ gives the conditions in canonical form
\begin{equation}
    T(u,0) = \frac{u}{4} \qquad \text{and} \qquad T(0,v) = -i^{2N-1} \frac{\beta}{\mu^{2N-2}} v^{2N-1} \, .
\end{equation}
Assuming an analytic solution \eqref{eq:tkesolncanon} to the time kernel equation \eqref{eq:tkecanon}, then we get conditions
\begin{equation} \label{eq:tkehnbca}
    \alpha_{m,0} = \frac{1}{4} \delta_{m,1} \qquad\text{and}\qquad \alpha_{0,n} = -i^{2N-1} \frac{\beta}{\mu^{2N-2}} \delta_{n,2N-1} \, ,
\end{equation}
and $\alpha_{m,n} = 0$ for negative $m$ or $n$.

\paragraph{Free Particle} For $V(q) = 0$, the boundary conditions give the solution
\begin{equation}
    T(q,q') = T_{\text{TOA}}(q,q') - i^{2N-1} \frac{\beta}{\mu^{2N-2}} (q-q')^{2N-1} \, ,
\end{equation}
which is just the time of arrival solution plus a $T_C$ term satisfying \eqref{eq:tkemgbc}. Using the Wigner-Weyl transform \eqref{eq:wignerweyl} and relation \eqref{eq:fouriertnsgnt}, we get
\begin{equation}
    \mathcal{T}_\hbar = -\frac{\mu q}{p} + \beta \frac{(2N-1)!}{2^{N-1}} \frac{\hbar^{2N-1}}{\mu^{3(N-1)}} \frac{1}{H^N} \, . 
\end{equation}
Note that second term of the transform, $\hbar^{2N-1} \mu^{-3(N-1)} H^{-N}$, also has the correct units of time. This gives us the free time of arrival shifted by a term proportional to $H^{-N}$, where $H$ is the classical Hamiltonian. It follows that the corresponding kernel solution $T_C$ is its supraquantization, for integer $N \geq 1$. This disappears in the classical limit unless $\beta$ is sufficiently large, like when $\beta \propto \hbar^{-(2N-1)}$.

\paragraph{Harmonic Oscillator} For $V(q) = \frac{1}{2} \mu \omega^2 q^2$, the boundary conditions give a solution
\begin{equation}
    T(q,q') = T_{\text{TOA}} -i^{2N-1}\frac{\beta}{\mu^{2N-2}} \sum_{s=0}^\infty \qty(\frac{\mu\omega}{2\hbar})^{2s} \frac{1}{2^s s!} \frac{\Gamma\qty(N+\frac{1}{2})}{2^s \Gamma\qty(N+\frac{1}{2}+s)} (q+q')^{2s} (q-q')^{2N-1+2s} \, .
\end{equation}
After using \eqref{eq:wignerweyl} and \eqref{eq:fouriertnsgnt}, 
\begin{equation}
    \mathcal{T}_\hbar = t_{\text{TOA}}(q,p) + \beta \frac{(2N-1)!}{2^{N-1}} \frac{\hbar^{2N-1}}{\mu^{3(N-1)}} \qty(\frac{2\mu}{p^2})^N \sum_{s=0}^\infty (-1)^s \frac{\Gamma(N+s)}{s! \Gamma(N)} \qty(\mu^2 \omega^2 \frac{q^2}{p^2})^s \, .
\end{equation}
We see that the second term is proportional to the series expansion of $H^{-N}$ for the harmonic oscillator. Thus, we get the harmonic oscillator time of arrival shifted by a term proportional to $H^{-N}$, where $H$ is the harmonic oscillator Hamiltonian, for $q$ sufficiently close to the origin.

Therefore, the boundary conditions \eqref{eq:tkehnbc} correspond to time kernel solutions that are the supraquantization of the series expansion of the time of arrival at the origin shifted by a term proportional to $H^{-N}$. The constructed supraquantized operators are unique, are Hermitian, and are conjugate to the Hamiltonian. The form of the supraquantized $H^{-N}$ is
\begin{equation}
    \beta \frac{(2N-1)!}{2^{N-1}} \frac{\hbar^{2N-1}}{\mu^{3(N-1)}} \braket{q | \mathcal{T}_C | q'} = \frac{\mu}{i\hbar} T_C(q,q') \sgn(q-q') \, ,
\end{equation}
where $T_C(q,q')$ is the shifting term of the time of arrival solution.

\paragraph{General Potential} We shall now look at the general potential case,
\begin{equation} \label{eq:genpol}
    V(q) = \sum_{s=1}^\infty a_s q^s \, .
\end{equation}
With this potential, the time kernel equation gives the recurrence relation \cite{galapon2004}
\begin{equation} \label{eq:genpolrecur}
    \alpha_{m,n} = \frac{\mu}{2\hbar^2} \frac{1}{mn} \sum_{s=1}^\infty \frac{a_s}{2^{s-1}} \sum_{k=0}^{[s]} \binom{s}{2k+1} \alpha_{m-s+2k,n-2k-2} \, ,
\end{equation}
where $[s] = \frac{s}{2} - 1$ for even $s$, and $[s] = \frac{s-1}{2}$ for odd $s$. The first boundary condition in \eqref{eq:tkehnbca} just gives the time of arrival solution. The second boundary condition in \eqref{eq:tkehnbca} spits out the shifting term $T_C$. 

We then focus on the $\alpha_{m,n}$'s generated by the second condition of \eqref{eq:tkehnbca}. For the general potential, we have, for positive $m$ and $j$ (see Appendix \ref{app:genpolsoln}),
\begin{equation}
    \alpha_{m,2N-1+2j} = \sum_{s=0}^{j-1} \qty(\frac{\mu}{2\hbar^2})^{j-s} \alpha_{m,j}^{(s)} \, ,
\end{equation}
where
\begin{equation}
    \alpha_{m,j}^{(s)} = \frac{1}{m(2N-1+2j)} \sum_{r=0}^{s} \sum_{n=2r+1}^{m+2r} \frac{a_n}{2^{n-1}} \binom{n}{2r+1} \alpha^{(s-r)}_{m-n+2r,j-r-1} \, .
\end{equation}
The Wigner-Weyl transform \eqref{eq:wignerweyl} gives
\begin{equation}
    \mathcal{T}_\hbar \propto \sum_{s=0}^{j-1} (-1)^j \qty(\frac{\mu}{2})^{j-s} (2N-1+2j)! \alpha^{(s)}_{m,j} \frac{\hbar^{2N-1+2s}}{p^{2j}} \, .
\end{equation}
We are only interested in the leading term $s = 0$; we ignore the higher order terms $s > 0$ for now. This leading term is of $\mathcal{O}(\hbar^{2N-1})$, and the $s > 0$ terms are of order $\mathcal{O}(\hbar^{2N+1})$. We then deal with
\begin{equation}
    \alpha^{(0)}_{m,j} = -i^{2N-1} \frac{\beta}{\mu^{2N-2}} \frac{\Gamma\qty(N+\frac{1}{2})}{\Gamma\qty(N+\frac{1}{2}+j)} \frac{1}{2^m} C_{m,j} \, ,
\end{equation}
where
\begin{equation}
    C_{m,j} = \frac{1}{m} \sum_{s=1}^m s a_s C_{m-s,j-1} \, ,
\end{equation}
and $C_{m,0} = \delta_{m,0}$. Then, our time kernel solution is of the form
\begin{equation}
    T^{(0)}(u,v) = T_{\text{TOA}}^{(0)}(u,v) - i^{2N-1} \frac{\beta}{\mu^{2N-2}} \sum_{j=0}^\infty \sum_{m=0}^\infty \qty(\frac{\mu}{2\hbar^2})^j \frac{\Gamma\qty(N+\frac{1}{2})}{\Gamma\qty(N+\frac{1}{2}+j)} \frac{1}{2^m} C_{m,j} u^m v^{2N-1+2j} \, ,
\end{equation}
where $T_{\text{TOA}}^{(0)}(u,v)$ is the leading term of the time of arrival solution generated by the first condition of \eqref{eq:tkehnbca}. The second term meanwhile is the leading term of the solution generated by the second condition in \eqref{eq:tkehnbca}. The Wigner-Weyl transform then gives
\begin{equation}
\begin{split}
    \mathcal{T}^{(0)}_\hbar(q,p) &= t_{\text{TOA}}(q,p) + \beta \frac{(2N-1)!}{2^{N-1}} \frac{\hbar^{2N-1}}{\mu^{3(N-1)}} \\
    &\qquad \qquad \qquad \quad \times \qty(\frac{2\mu}{p^2})^N \sum_{k=0}^\infty (-1)^k \frac{\Gamma(N+k)}{k!\Gamma(N)} \qty(\frac{2\mu}{p^2})^k \qty(\sum_{s=1}^\infty a_s q^s)^k \, ,
\end{split}
\end{equation}
where the second term is nothing but a term proportional to $H^{-N}$ for the general potential.

Thus, even in the general potential case, we still get the time of arrival shifted by a term proportional to an inverse power of the Hamiltonian. We then have the following Proposition.

\begin{proposition}
For an entire analytic potential $V(q)$, and for constant $\beta \in \mathbb{R}$ and $N \in \mathbb{Z}^+$, the solution to the time kernel equation
\begin{equation}
    -\frac{\hbar^2}{2\mu} \, \pdv[2]{T(q,q')}{q} + \frac{\hbar^2}{2\mu} \, \pdv[2]{T(q,q')}{{q'}} + [V(q) - V(q')] T(q,q') = 0  \, ,
\end{equation}
with boundary conditions
\begin{equation}
    T(q,q) = \frac{q}{2} \qquad T(q,-q) = - i^{2N-1} \beta \mu^{-2(N-1)} (2q)^{2N-1} \, ,
\end{equation}
is given by
\begin{equation}
    T(q,q') = T_\text{TOA}(q,q') + T_C(q,q') \, ,
\end{equation}
where the first term $T_\text{TOA}(q,q')$ is the kernel of the time of arrival operator, and the second term $T_C(q,q')$ gives the Wigner-Weyl transform
\begin{equation} \label{eq:wwtkehnbc}
    \mathcal{T}_{\hbar,C}(q,p) = 
\begin{cases}
    \beta \frac{(2N-1)!}{2^{N-1}} \frac{\hbar^{2N-1}}{\mu^{3(N-1)}} \qty(\frac{1}{H(q,p)})^N \, , &\qquad \text{for linear systems,} \\
    \beta \frac{(2N-1)!}{2^{N-1}} \frac{\hbar^{2N-1}}{\mu^{3(N-1)}} \qty(\frac{1}{H(q,p)})^N + \order{\hbar^{2N+1}} \, , &\qquad \text{for nonlinear systems.}
\end{cases}
\end{equation}
\end{proposition}

In the classical limit $\hbar \to 0$, all that remains in the Wigner-Weyl of $T(q,q')$ is the local time of arrival (at the origin) term $t_0(q,p)$. If $\beta \propto 1/\hbar^{2N-1}$, then the higher order terms in \eqref{eq:wwtkehnbc} reduce from $\order{\beta\hbar^{2N+1}}$ to $\order{\hbar^2}$, and the classical limit is then the local time of arrival shifted by a term proportional to a negative power of the Hamiltonian. We then interpret the time kernel solution 
\begin{equation} \label{eq:tkehnbcgen}
    \braket{q|\mathcal{T}|q'} = \braket{q|\mathcal{T}_\text{TOA}|q'} + \beta \frac{(2N-1)!}{2^{N-1}} \frac{\hbar^{2N-1}}{\mu^{3(N-1)}} \braket{q | \mathcal{T}_C | q'} \, ,
\end{equation}
as the supraquantization of 
\begin{equation}
    \mathcal{T}_0(q,p) = t_0(q,p) + \beta \frac{(2N-1)!}{2^{N-1}} \frac{\hbar^{2N-1}}{\mu^{3(N-1)}} \qty(\frac{1}{H(q,p)})^N \, .
\end{equation}
Again, we note that $\braket{q | \mathcal{T}_C | q'}$ and $\braket{q | \operator{H}^{-N} | q'}$ are not equal in general, since only the leading term of $\braket{q | \mathcal{T}_C | q'}$ corresponds to $H^{-N}$.

\subsection{An example satisfying time-reversal symmetry}
The inverse powers of the Hamiltonian constitute solutions that can be Hermitian, but not satisfy time reversal symmetry. The second boundary condition for the Hamiltonian solution was in the form $\beta_{2N-1} q^{2N-1}$ for positive integer $N$, i.e., they only have nonvanishing $\beta_k$'s for odd $k$. To satisfy time-reversal symmetry, $\beta_k$ should vanish for odd $k$ and be purely real for even $k$.

Suppose we have the boundary condition
\begin{equation}
    T(q,q) = \frac{q}{2} \qquad\text{and}\qquad T(q,-q) = \lambda q^2 \, ,
\end{equation}
where $\lambda$ is some constant with units of inverse length. This gives, for the free particle case $V(q) = 0$, the time kernel solution
\begin{equation}
    T(q,q') = \frac{1}{4}(q+q') + \lambda \qty(\frac{q-q'}{2})^2 \, ,
\end{equation}
which gives
\begin{equation}
    \braket{q|\mathcal{T}|q'} = \frac{\mu}{4i\hbar} (q+q') \sgn(q-q') + \lambda \frac{\mu}{4i\hbar} (q-q')^2 \sgn(q-q') \, .
\end{equation}
This kernel is Hermitian and satisfies time reversal symmetry. The Wigner-Weyl transform gives
\begin{equation}
    \mathcal{T}_\hbar(q,p) = -\frac{\mu q}{p} + \lambda \frac{\mu \hbar^2}{p^3} \, .
\end{equation}
The second term then has units of time, and is also interpreted to shift the time of arrival result. This shifting term vanishes as $\hbar \to 0$. Note that this shifting term is not anymore a function of the Hamiltonian but still corresponds to an operator which commutes with the Hamiltonian. This demonstrates that the term added to a Hamiltonian conjugate need not be a function of a Hamiltonian to still satisfy the canonical commutation relation.

\section{The Modified Time Kernel Equation} \label{sec:mtke}
The transfer principle can be exhibited in a less explicit way in deriving the time of arrival operator \cite{domingo2004}. This gives a more general way of establishing quantum-classical correspondence via supraquantization. Instead of assuming a definite form of the kernel, the whole kernel $\braket{q | \mathcal{T} | q'}$ can be regarded as unknown. We then interpret our time operator $\mathcal{T}$ as a distribution on some test function space, i.e., $\braket{\mathcal{T},\varphi} = \int \mathcal{T}(q,q') \varphi(q,q') \dd{q} \dd{q'}$.

Since $\rhsextension{H}$ is a mapping from $\dual{\Phi}$ to $\dual{\Phi}$, then $\rhsextension{H}\mathcal{T}$ is well defined for $\mathcal{T}\varphi \in \dual{\Phi}$. For the $\mathcal{T}\rhsextension{H}$ term, we choose $\Phi$ to be invariant under $\rhsextension{H}$ (which is satisfied if $\Phi$ is invariant under a self-adjoint $\operator{H}$); here we continue choosing $\Phi$ to be the space of infinitely differentiable functions with compact support. We then have, for $\varphi$ and $\tilde{\varphi}$ in $\Phi$,
\begin{equation}
\begin{split}
\braket{ \tilde{\varphi} | [ \rhsextension{H}, \mathcal{T} ] \varphi } &= \iint \tilde{\varphi}^*(q) \qty[-\frac{\hbar^2}{2\mu} \, \pdv[2]{ \mathcal{T}(q,q')}{q} + V(q)\mathcal{T}(q,q')] \varphi(q')\dd{q'}\dd{q} \\
& \quad - \iint \tilde{\varphi}^*(q) \qty[-\frac{\hbar^2}{2\mu} \, \pdv[2]{ \mathcal{T}(q,q')}{{q'}} + V(q')\mathcal{T}(q,q')] \varphi(q')\dd{q'}\dd{q} \, ,
\end{split}
\end{equation}
where $\mathcal{T}(q,q')$ is essentially the kernel $\braket{q|\mathcal{T}|q'}$ in the original time kernel equation; and where the $\mathcal{T}\rhsextension{H}$ is obtained using the definition of the derivative of a distribution $\mathcal{T} \varphi^{(n)} = (-1)^n \mathcal{T}^{(n)} \varphi$, thus
\begin{equation}
    \mathcal{T}\rhsextension{H}\varphi(q') = \mathcal{T}\qty(-\frac{\hbar^2}{2\mu} \dv[2]{{q'}}\varphi(q') + V(q')\varphi(q')) = -\frac{\hbar^2}{2\mu} \dv[2]{\mathcal{T}}{{q'}}\varphi(q') + V(q')\mathcal{T}\varphi(q') \, .
\end{equation}
This holds for any $\Phi$ where the derivative of a distribution is related to the distribution of the derivative in this way. Note that the integrals here are to be understood in the distributional sense. They reduce to the usual integration for regular distributions, but is only symbolic for singular distributions (e.g., Dirac delta distribution).

If we then impose the canonical commutation relation $[ \rhsextension{H}, \mathcal{T} ] = i\hbar$, we must require that
\begin{equation} \label{eq:mtke}
	-\frac{\hbar^2}{2\mu} \, \pdv[2]{ \mathcal{T}(q,q')}{q} + \frac{\hbar^2}{2\mu} \, \pdv[2]{ \mathcal{T}(q,q')}{{q'}} + [V(q) - V(q')]\mathcal{T}(q,q') = i\hbar\delta(q-q') \, .
\end{equation}
We call \eqref{eq:mtke} the \textbf{modified time kernel equation}. Note that \textit{any} solution to the modified time kernel equation \eqref{eq:mtke} automatically corresponds to an operator that is canonically conjugate with the Hamiltonian. For the boundary condition $\mathcal{T}(\Gamma) = 0$ for any boundary $\Gamma$, it was shown to have the same result as the original time kernel equation for linear systems \cite{domingo2004}. Comparing with the original time kernel equation \eqref{eq:tke} and boundary condition \eqref{eq:tkegbc}, the modified time kernel equation is written as one equation instead of two, and does not have the assumption that $\mathcal{T}(q,q') = (\mu/i\hbar)T(q,q')\sgn(q-q')$ for analytic $T(q,q')$.

We first express \eqref{eq:mtke} in its canonical form by changing variables to $u = q + q'$ and $v = q - q'$,
\begin{equation} \label{eq:mtkecanon}
    -\frac{2\hbar^2}{\mu} \, \pdv[2]{\mathcal{T}(u,v)}{u}{v} + \qty[V\qty(\frac{u+v}{2}) - V\qty(\frac{u-v}{2})] \mathcal{T}(u,v) = i\hbar \delta(v)  \, .
\end{equation}
In \cite{domingo2004}, given the time of arrival boundary conditions, \eqref{eq:mtkecanon} was solved using Green's function and method of successive approximation. Here, to get the general solution, we first note that
\begin{equation}
    \int \delta(y) \dd{y} = \alpha H(v) - \beta H(-v) + c \, ,
\end{equation}
where $H(x)$ is the Heaviside function, $\alpha$ and $\beta$ are constants, $c$ is the constant distribution, and $\alpha + \beta = 1$. Note that the antiderivative of $\delta$ still resides inside $\dual{\Phi}$.

From \eqref{eq:mtkecanon}, we then get the integral form of the modified time kernel equation,
\begin{equation}
\begin{split} \label{eq:mtkeint}
    \mathcal{T}(u,v) &= \frac{\mu}{2 i \hbar} u [\alpha H(v) - \beta H(-v)] + f(u) + g(v) \\
    &\quad + \frac{\mu}{2\hbar^2} \int_0^v \int_0^u \qty[V\qty(\frac{x+y}{2}) - V\qty(\frac{x-y}{2})] \mathcal{T}(x,y) \dd{x} \dd{y} \, ,
\end{split}
\end{equation}
where $f$ and $g$ are distributions in $u$ and $v$, respectively. The distribution $\mathcal{T}$ acts on test functions of two independent variables $\phi(u,v) \in \Phi$. Note that here, $\frac{\partial^{k_1+k_2}}{\partial u^{k_1} \partial v^{k_2}} \phi(u,v)$ exists and is continuous everywhere for all positive integers $k_1$ and $k_2$. Also, the support of $\phi$ is the closure of the set of all points in $(u,v)$ wherein $\phi(u,v)$ is nonzero.

The definite integrals in \eqref{eq:mtkeint} will depend on the nature of the potential $V(q)$ and the time kernel distribution $\mathcal{T}$. For now, we will only deal with the case of continuous potentials and with regular distributions $\mathcal{T}$, i.e., the $f(u)$ and $g(v)$ are locally integrable functions. One could interpret the integrals as
\begin{equation}
\begin{split}
    &\Braket{\qty[V\qty(\frac{x+y}{2}) - V\qty(\frac{x-y}{2})] \mathcal{T}(x,y) \theta(x,y), \lambda(x,y)} \\
    &\qquad\qquad\qquad\qquad\qquad\qquad = \int_0^v \int_0^u \qty[V\qty(\frac{x+y}{2}) - V\qty(\frac{x-y}{2})] \mathcal{T}(x,y) \dd{x} \dd{y} \, ,
\end{split}
\end{equation}
where $\theta(x,y)$ is identically equal to one over the region $0 \leq x \leq u$ and $0 \leq y \leq v$ (c.f. over the region $\mathcal{R}$ in Figure \ref{fig:region}) and is zero outside, and where $\lambda(x,y) \in \Phi$ is identically equal to one over neighborhood of the region $0 \leq x \leq u$ and $0 \leq y \leq v$, and is zero outside some bigger region.

Finally, we note that two distributions $\mathcal{T}_1$ and $\mathcal{T}_2$ are equal if
\begin{equation}
    \braket{\mathcal{T}_1,\phi} = \braket{\mathcal{T}_2,\phi} \, ,
\end{equation}
for every $\phi \in \Phi$; in other words, $\braket{\mathcal{T}_1 - \mathcal{T}_2,\phi} = 0$. If these regular distributions correspond to continuous functions, then $\mathcal{T}_1$ and $\mathcal{T}_2$ must be identical (at least, in the neighborhood of the support of $\phi$). In general, the regular distributions, $\mathcal{T}_1$ and $\mathcal{T}_2$, correspond to locally integrable functions, which include functions with points where it is discontinuous (e.g. the Heaviside function). This means that two equal regular distributions can differ at most on a set of measure zero \cite{zemanian1987}.

Below, we will show that there exists a regular distribution $\mathcal{T}$ that is a solution to the modified time kernel equation. The proof is analogous to Theorem \ref{thm:tkeunique}, where here, we recast it in a distributional sense.

\begin{theorem} \label{thm:mtkeexist}
Consider the integral form of the modified time kernel equation given by \eqref{eq:mtkeint}. If $V(q)$ is a continuous function at any point of the real line, and if $f$ and $g$ are locally integrable functions, then there exists a regular distribution $\mathcal{T}$ that is a solution to the modified time kernel equation \eqref{eq:mtke}.
\end{theorem}

\begin{proof}
From the integral form of the modified time kernel equation \eqref{eq:mtkeint}, we can use the method of successive approximation. Let $\qty{\mathcal{T}_n}_{n=0}^\infty$ be a sequence of locally integrable functions where
\begin{equation}
    \mathcal{T}_0(u,v) = \frac{\mu}{2 i \hbar} u [\alpha H(v) - \beta H(-v)] + f(u) + g(v) \, ,
\end{equation}
and
\begin{equation}
\begin{split} 
    \mathcal{T}_n(u,v) &= \frac{\mu}{2 i \hbar} u [\alpha H(v) - \beta H(-v)] + f(u) + g(v) \\
    &\quad + \frac{\mu}{2\hbar^2} \int_0^v \int_0^u \qty[V\qty(\frac{x+y}{2}) - V\qty(\frac{x-y}{2})] \mathcal{T}_{n-1}(x,y) \dd{x} \dd{y} \, .
\end{split}
\end{equation}
For locally integrable functions $f(u)$ and $g(v)$, the first approximation $\mathcal{T}_0(u,v)$ is also a locally integrable function. It follows that $\mathcal{T}_n(u,v)$ is also locally integrable. They then correspond to regular distributions in $\dual{\Phi}$. We would like to show that $\qty{\mathcal{T}_n(u,v)}_{n=0}^\infty$ converges. 

It is convenient to also write $\mathcal{T}_n$ as
\begin{equation} \label{eq:mtketnsum}
    \mathcal{T}_n(u,v) = \frac{\mu}{2 i \hbar} u [\alpha H(v) - \beta H(-v)] + f(u) + g(v) + \sum_{j=1}^n \qty[\mathcal{T}_j(u,v) - \mathcal{T}_{j-1}(u,v)] \, .
\end{equation}
Since $V(q)$ is continuous, there exists an $M > 0$ such that $\abs{V\qty(\frac{x+y}{2}) - V\qty(\frac{x-y}{2})} \leq M$ for any point $(x,y)$ in the region $\mathcal{R}$ in Figure \ref{fig:region}. Also, since $f$ and $g$ are locally integrable, then there exists an $N_f > 0$ and $N_g > 0$ such that $\int_{\Omega_f} \abs{f(x)} \dd{x} \leq N_f$ and $\int_{\Omega_g} \abs{g(y)} \dd{y} \leq N_g$ for all compact subsets $\Omega_f \subset \domain{f}$ and $\Omega_g \subset \domain{g}$, where $\domain{f}$ and $\domain{g}$ are the domains of $f$ and $g$ respectively. We then have,
\begin{align}
    \abs{\mathcal{T}_1(u,v) - \mathcal{T}_0(u,v)} &\leq \frac{\mu M}{2 \hbar^2} \int_0^v \int_0^u \abs\bigg{\frac{\mu}{2 i \hbar} x \qty\Big[\alpha H(y) - \beta H(-y)] + f(x) + g(y)} \dd{x} \dd{y} \nonumber \\
    &\leq \frac{\mu M}{2 \hbar^2} \int_0^v \int_0^u \qty\bigg[\frac{\mu}{2 \hbar} \abs{x} \qty\Big( \abs{\alpha} \abs{H(y)} + \abs{\beta} \abs{H(-y)} ) + \abs{f(x)} + \abs{g(y)}] \dd{x} \dd{y} \nonumber \\
    &\leq \frac{\mu M}{2 \hbar^2} \qty[ \frac{\mu}{2\hbar} \frac{\abs{u}^2 \abs{v}}{2} \qty\Big(\abs{\alpha} + \abs{\beta}) + \abs{u}\abs{v} \qty\Big(N_f + N_g) ] \nonumber \\
    &\quad = \frac{\mu M}{2 \hbar^2} \abs{u}\abs{v} \qty[\frac{\mu}{2\hbar} \frac{\abs{u}}{2} \qty\Big(\abs{\alpha} + \abs{\beta}) + \qty\Big(N_f + N_g) ] \, .
\end{align}
Similarly,
\begin{align}
    \abs{\mathcal{T}_2(u,v) - \mathcal{T}_1(u,v)} &\leq \frac{\mu M}{2 \hbar^2} \int_0^v \int_0^u \abs{\mathcal{T}_1(x,y) - \mathcal{T}_0(x,y)} \dd{x} \dd{y} \nonumber \\
    &\leq \qty(\frac{\mu M}{2 \hbar^2})^2 \frac{\abs{u}^2\abs{v}^2}{2 \cdot 2} \qty[\frac{\mu}{2\hbar} \frac{\abs{u}}{3} \qty\Big(\abs{\alpha} + \abs{\beta}) + \qty\Big(N_f + N_g) ] \, ,
\end{align}
and,
\begin{equation}
    \abs{\mathcal{T}_3(u,v) - \mathcal{T}_2(u,v)} \leq \qty(\frac{\mu M}{2 \hbar^2})^3 \frac{\abs{u}^3\abs{v}^3}{(3 \cdot 2)(3 \cdot 2)} \qty[\frac{\mu}{2\hbar} \frac{\abs{u}}{4} \qty\Big(\abs{\alpha} + \abs{\beta}) + \qty\Big(N_f + N_g) ] \, .
\end{equation}
By induction,
\begin{equation}
    \abs{\mathcal{T}_j(u,v) - \mathcal{T}_{j-1}(u,v)} \leq \qty(\frac{\mu M}{2 \hbar^2})^j \frac{\abs{u}^j\abs{v}^j}{(j!)^2} \qty[\frac{\mu}{2\hbar} \frac{\abs{u}}{(j+1)} \qty\Big(\abs{\alpha} + \abs{\beta}) + \qty\Big(N_f + N_g) ] \, .
\end{equation}
We see that this goes to zero as $j$ approaches infinity. Thus, the partial sum in \eqref{eq:mtketnsum} is absolutely and uniformly convergent for all finite values of $u$ and $v$. Therefore, the sequence $\qty{\mathcal{T}_n(u,v)}_{n=0}^\infty$ converges. 

Let the limit of $\mathcal{T}_n$ as $n \to \infty$ be denoted by $\mathcal{T}$. Then, $\mathcal{T}$ is a solution to the modified time kernel equation. Since $\dual{\Phi}$ is closed under convergence \cite{zemanian1987}, then $\mathcal{T}_n \in \dual{\Phi}$ implies that $\mathcal{T}$ is also in $\dual{\Phi}$, i.e., $\mathcal{T}$ is a regular distribution for locally integrable functions $f$ and $g$.
\end{proof}

When $\mathcal{T}$ is a regular distribution, then there is a one-to-one relation between $\mathcal{T}$ and a locally integrable function $\mathcal{T}(q,q') = \braket{q|\mathcal{T}|q'}$. The uniqueness proof of Theorem \ref{thm:tkeunique} could also be attempted here. One has to note that all the locally integrable functions which differ at most on a set of measure zero produce the same regular distribution. Each solution can then be thought of as an equivalence class of functions corresponding to a particular regular distribution. Additionally, Theorem \ref{thm:mtkeexist} only shows the existence of a distributional solution for locally integrable function $f$ and $g$. This does not encompass all the possible solutions, since one could also choose $f$ and $g$ to be Dirac deltas, making $\mathcal{T}$ a singular distribution.

\paragraph{Hermiticity} Our operator $\mathcal{T}$ is Hermitian if $\mathcal{T}(q,q') = \mathcal{T}^*(q',q)$; in canonical form, 
\begin{equation}
    \mathcal{T}(u,v) = \mathcal{T}^*(u,-v) \, .
\end{equation}
From the general solution \eqref{eq:mtkeint}, we see that
\begin{equation}
\begin{split}
    \mathcal{T}^*(u,-v) &= \frac{\mu}{2 i \hbar} u [\beta^* H(v) - \alpha^* H(-v)] + f^*(u) + g^*(-v) \\
    &\quad + \frac{\mu}{2\hbar^2} \int_0^v \int_0^u \qty[V\qty(\frac{x+y}{2}) - V\qty(\frac{x-y}{2})]^* \mathcal{T}^*(x,-y) \dd{x} \dd{y} \, .
\end{split}
\end{equation}
Comparing the first terms of $\mathcal{T}(u,v)$ and $\mathcal{T}^*(u,-v)$ gives
\begin{equation}
    (\alpha - \beta^*)H(v) = (\beta - \alpha^*)H(-v) \, .
\end{equation}
Since $H(v)$ and $H(-v)$ are independent, then we get
\begin{equation}
    \alpha = \beta^* \, .
\end{equation}
With $\alpha + \beta = 1$, then the above condition also implies that $\alpha + \alpha^* = 1$, and similarly, $\beta + \beta^* = 1$.

Thus, the Hermiticity condition $\mathcal{T}(u,v) = \mathcal{T}^*(u,-v)$ is satisfied if
\begin{equation}
    \Re(\alpha) = \Re(\beta) = \frac{1}{2} \, ,
\end{equation}
\begin{equation}
    \alpha = \beta^* \, ,
\end{equation}
\begin{equation}
    f(u) = f^*(u) \, ,
\end{equation}
\begin{equation}
    g(v) = g^*(-v) \, ,
\end{equation}
\begin{equation}
    V(q) = V^*(q) \, .
\end{equation}
If the real parts of $\alpha$ and $\beta$ are $1/2$, then $\mathcal{T}(u,v)$ will contain a term with $\sgn(v)$, which will generate the usual time of arrival solution.

\paragraph{Time-reversal symmetry} Our operator $\mathcal{T}$ satisfies time-reversal symmetry if $\mathcal{T}^*(q,q') = -\mathcal{T}(q,q')$; in canonical form,
\begin{equation}
    \mathcal{T}^*(u,v) = -\mathcal{T}(u,v) \, .
\end{equation}
Then, comparing the first terms give
\begin{equation}
    (\alpha - \alpha^*) H(v) = (\beta - \beta^*) H(-v) \, .
\end{equation}
Again, since $H(v)$ and $H(-v)$ are independent, then $\alpha = \alpha^*$ and $\beta = \beta^*$.

Thus, time-reversal symmetry holds when
\begin{equation}
    \alpha = \alpha^* \, , \qquad \beta = \beta^* \, ,
\end{equation}
\begin{equation}
    \alpha + \beta = 1 \, ,
\end{equation}
\begin{equation}
    f(u) = -f^*(u) \, ,
\end{equation}
\begin{equation}
    g(v) = -g^*(v) \, ,
\end{equation}
\begin{equation}
    V(q) = V^*(q) \, .
\end{equation}
In other words, $\alpha$, $\beta$, and $V(q)$ are purely real while $f(u)$ and $g(v)$ are purely imaginary.

\paragraph{Both Hermitian and Time-reversal symmetric} If we want our solution to satisfy both Hermiticity and time-reversal symmetry, then the conditions for both must be simultaneously imposed. For the constants $\alpha$ and $\beta$, since time-reversal symmetry requires purely real $\alpha$ and $\beta$, and Hermiticity requires that the real part of both be $1/2$; thus $\alpha = \beta = 1/2$. For $f$, we have the condition $f(u) = f^*(u) = -f^*(u)$, which gives $f^*(u) = 0$, and thus, $f(u) = 0$. For $g$, we have $g(v) = g^*(-v) = -g^*(v)$, which means that $g$ is purely imaginary and odd. Thus, both Hermiticity and time-reversal symmetry are satisfied if
\begin{equation}
    \alpha = \beta = \frac{1}{2} \, ,
\end{equation}
\begin{equation}
    f(u) = 0 \, ,
\end{equation}
\begin{equation}
    g(v) = g^*(-v) = -g^*(v) \, ,
\end{equation}
\begin{equation}
    V(q) = V^*(q) \, .
\end{equation}

The $\alpha = \beta = 1/2$ condition makes the first term in \eqref{eq:mtkeint} contain a $\sgn(v)$ which gives the time of arrival solution. We let $f$ vanish, $g(v)$ be odd and purely imaginary, and the potential $V(q)$ to be purely real.

\subsection{Free Particle}
For the free particle case $V(q) = 0$, the general solution is simply
\begin{equation} \label{eq:mtkefree}
    \mathcal{T}(u,v) = \frac{\mu}{2 i \hbar} u [\alpha H(v) - \beta H(-v)] + f(u) + g(v) \, .
\end{equation}
Using the Wigner-Weyl transform \eqref{eq:wignerweyl} and using the relation for the Fourier transform of $H(t)$, $\fourier H(t)$, where $\fourier f = \int_{-\infty}^\infty f(t) e^{-i\omega t} \dd{t}$, we have \cite{zemanian1987},
\begin{equation}
    \int_{-\infty}^\infty H(t) e^{-i t\omega} \dd{t} = \pi \delta(\omega) + \frac{1}{i\omega} \, .
\end{equation}
Using $\fourier 1 = 2\pi \delta(\omega)$ as well, we then get
\begin{equation}
\begin{split}
    \mathcal{T}_\hbar(q,p) &= -\frac{\mu q}{p} - i \pi \mu q (\alpha - \beta) \delta(p) \\
    & \quad + 2\pi \hbar f(2q) \delta(p) + \int_{-\infty}^\infty g(v) \exp\qty(-i\frac{p}{\hbar}v) \dd{v} \, .
\end{split}
\end{equation}
The first term is nothing but the free time of arrival at the origin. Even without requiring Hermiticity or time-reversal symmetry, being conjugate to the Hamiltonian alone brought upon this time of arrival term. Conditions for $\alpha$, $\beta$, $f$, and $g$ then determine the shift from the usual time of arrival result---which, in essence, determine all the other Hamiltonian conjugate solutions.

One could observe that the modified time kernel equation gives a more general solution compared to the original time kernel. For example, If $\alpha \neq \beta$ and $f(u) \neq 0$, then the Wigner-Weyl transform will have terms containing Dirac deltas. Since the time of arrival term $-\mu q/p$ is only valid for $p \neq 0$, we interpret the delta term as a contribution of a stationary particle, i.e., when $p=0$, the particle will never reach the arrival point unless it is already there; we must necessarily require that $f(0) = 0$ in this interpretation. This ``stationary particle" term will disappear in the classical limit $\hbar \to 0$ unless $f$ is sufficiently large, e.g., $f \propto \hbar^{-1}$.

These Dirac delta terms cannot appear in the original time kernel equation solution, since the form of the kernel is always assumed to have $\sgn(v)$ multiplied by an analytic function in $v$. In using the transform \eqref{eq:wignerweyl} and the Fourier transform of $v\sgn(v)$, we see that a Dirac delta can never arise in the classical limit of the original kernel. Therefore, the original free time kernel equation solutions only correspond to moving particles. Meanwhile, the free modified time kernel equation allows stationary particles.

Let us now look at the consequences of imposing Hermiticity and time-reversal symmetry on the Wigner-Weyl transform. 
\begin{itemize}
    \item Hermiticity only: The second term containing $i(\alpha - \beta) = -2\Im(\alpha)$ will become purely real, and the third term containing $f$ will be purely real. The condition on $g$ will ensure that the Fourier transform is real ($g(v) = g^*(-v)$ implies that $\fourier g = (\fourier g)^*$). The Wigner-Weyl transform will be purely real and will continue to contain Dirac delta terms.
    
    \item Time-reversal symmetry only: The second and third term is purely imaginary. The Fourier transform of $g$ is not necessarily real anymore. Dirac deltas are still present.
    
    \item Both Hermiticity and time-reversal symmetry: Both the second and third term vanish, and the last term is purely real. Here, all the Dirac delta terms have vanished.
\end{itemize}

Note that by letting $\alpha = \beta = 1/2$ and $f(u) = g(v) = 0$, we get the free classical time of arrival (at the origin) $\mathcal{T}_\hbar(q,p) = -\mu q p^{-1}$, and the corresponding free time kernel $\mathcal{T}(u,v) = \mu (4 i \hbar)^{-1} u \sgn(v)$. 

The term containing $g$ is essentially a distributional Fourier transform, and its form will depend on $g(v)$. One interesting example is by setting $g(v) = - \mu \hbar^{-1} v \sgn(v)$. The Wigner-Weyl transform of this term is $\hbar/H(q,p)$, where $H(q,p)$ is the free classical Hamiltonian. This gives the inverse Hamiltonian shifted time of arrival, which we also constructed with the original time kernel equation.

\subsection{Harmonic Oscillator}
In general, for $V(q) \neq 0$, the solution can be constructed using an iterative solution to the Fredholm integral of the second kind. For the harmonic oscillator $V(q) = \mu\omega^2 q^2/2$, we have
\begin{equation}
\begin{split}
    \mathcal{T}(u,v) &= \frac{\mu}{2 i \hbar} u [\alpha H(v) - \beta H(-v)] + f(u) + g(v) \\
    &\quad + \frac{\mu^2 \omega^2}{4\hbar^2} \int_0^v \int_0^u x y \mathcal{T}(x,y) \dd{x} \dd{y} \, .
\end{split}
\end{equation}
Let the initial approximation $\mathcal{T}_0$ be 
\begin{equation}
    \mathcal{T}_0(u,v) = \frac{\mu}{2 i \hbar} u [\alpha H(v) - \beta H(-v)] + f(u) + g(v) \, ,
\end{equation}
and let $\mathcal{T}_n$ be
\begin{equation}
\begin{split}
    \mathcal{T}_n(u,v) &= \frac{\mu}{2 i \hbar} u [\alpha H(v) - \beta H(-v)] + f(u) + g(v) \\
    &\quad + \frac{\mu^2 \omega^2}{4\hbar^2} \int_0^v \int_0^u xy \mathcal{T}_{n-1}(x,y) \dd{x} \dd{y} \, .
\end{split}
\end{equation}
Using the above equations, we get the next approximation $\mathcal{T}_1$, given by
\begin{equation}
\begin{split}
    \mathcal{T}_1 &= \mathcal{T}_0 + \frac{\mu}{2 i \hbar} \frac{\mu^2 \omega^2}{4\hbar^2} \int_0^v \int_0^u x^2 y [\alpha H(y) - \beta H(-y)] \dd{x} \dd{y} \\
    &\quad + \frac{\mu^2 \omega^2}{4\hbar^2} \int_0^v \int_0^u x y f(x) \dd{x} \dd{y} + \frac{\mu^2 \omega^2}{4\hbar^2} \int_0^v \int_0^u x y \, g(y) \dd{x} \dd{y} \, .
\end{split}
\end{equation}
Note that
\begin{equation}
    \int_0^v y [\alpha H(y) - \beta H(-y)] \dd{y} = 
    \begin{cases}
        \frac{v^2}{2} \alpha \, , & \quad \text{for } v > 0 \, , \\
        \frac{v^2}{2} (-\beta) \, , & \quad \text{for } v < 0 \, .
    \end{cases} 
\end{equation}
and so, we have
\begin{equation}
\begin{split}
    \mathcal{T}_1 &= \mathcal{T}_0 + \frac{\mu}{2 i \hbar} \frac{\mu^2 \omega^2}{4\hbar^2} \frac{u^3}{3} \frac{v^2}{2} [\alpha H(v) - \beta H(-v)] \\
    &\quad + \frac{\mu^2 \omega^2}{4\hbar^2} \frac{v^2}{2} \int_0^u x f(x) \dd{x} + \frac{\mu^2 \omega^2}{4\hbar^2} \frac{u^2}{2} \int_0^v y \, g(y) \dd{y} \, .
\end{split}
\end{equation}
The next approximation gives
\begin{equation}
\begin{split}
    \mathcal{T}_2 &= \mathcal{T}_1 + \frac{\mu}{2 i \hbar} \qty(\frac{\mu^2 \omega^2}{4\hbar^2})^2 \qty(\frac{u^5}{1 \cdot 3 \cdot 5}) \qty(\frac{v^4}{2 \cdot 4}) [\alpha H(v) - \beta H(-v)] \\
    &\quad + \qty(\frac{\mu^2 \omega^2}{4\hbar^2})^2 \frac{v^4}{2 \cdot 4} \int_0^u \int_0^x x' f(x') \dd{x'} \dd{x} + \qty(\frac{\mu^2 \omega^2}{4\hbar^2})^2 \frac{u^4}{2 \cdot 4} \int_0^v \int_0^y y' \, g(y') \dd{y'} \dd{y} \, .
\end{split}
\end{equation}
By induction, we get
\begin{equation}
\begin{split}
    \mathcal{T}_n(u,v) &= \frac{\mu}{2 i \hbar} \sum_{j=0}^n \qty(\frac{\mu\omega}{2\hbar})^{2j} \frac{1}{(2j+1)!} u^{2j+1} v^{2j} [\alpha H(v) - \beta H(-v)] \\
    &+ \sum_{j=0}^n \qty(\frac{\mu\omega}{2\hbar})^{2j} \frac{v^{2j}}{2^j j!} F_j(u) + \sum_{j=0}^n \qty(\frac{\mu\omega}{2\hbar})^{2j} \frac{u^{2j}}{2^j j!} G_j(v) \, ,
\end{split}
\end{equation}
for $n \geq 0$, where,
\begin{equation}
    F_0(u) = f(u) \, , \qquad G_0(v) = g(v) \, ,
\end{equation}
\begin{equation}
    F_s(u) = \int_0^u x F_{s-1}(x) \dd{x} \, , \qquad \text{and} \qquad G_s(v) = \int_0^v y G_{s-1}(y) \dd{y} \qquad \text{for } s \geq 1 \, .
\end{equation}
We arrive at the solution by letting $n \to \infty$,
\begin{equation}
\begin{split}
    \mathcal{T}(u,v) &= \frac{\mu}{2 i \hbar} \sum_{j=0}^\infty \qty(\frac{\mu\omega}{2\hbar})^{2j} \frac{1}{(2j+1)!} u^{2j+1} v^{2j} [\alpha H(v) - \beta H(-v)] \\
    &+ \sum_{j=0}^\infty \qty(\frac{\mu\omega}{2\hbar})^{2j} \frac{v^{2j}}{2^j j!} F_j(u) + \sum_{j=0}^\infty \qty(\frac{\mu\omega}{2\hbar})^{2j} \frac{u^{2j}}{2^j j!} G_j(v) \, .
\end{split}
\end{equation}
This is the general solution to the modified time kernel equation for the harmonic oscillator case.

Using $\fourier(it)^k = (-1)^k 2\pi \delta^{(k)}(\omega)$, and \cite{zemanian1987}
\begin{equation}
    \int_{-\infty}^\infty t^n H(t) e^{-it\omega} \dd{t} = i^n \pi \delta^{(n)}(\omega) + \frac{n!}{(i\omega)^{n+1}} \, ,
\end{equation}
the Wigner-Weyl transform gives
\begin{equation}
\begin{split}
    \mathcal{T}_\hbar(q,p) &= - \sum_{j=0}^\infty \frac{(-1)^j}{2j+1} \mu^{2j+1} \omega^{2j} \frac{q^{2j+1}}{p^{2j+1}} - i \pi (\alpha - \beta) \sum_{j=0}^\infty \frac{(-1)^j}{(2j+1)!} \frac{\mu^{2j+1} \omega^{2j}}{\hbar^{2j}} q^{2j+1} \delta^{(2j)}(p) \\
    &\quad + 2\pi \sum_{j=0}^\infty \frac{(-1)^j}{8^j j!} \frac{\mu^{2j} \omega^{2j}}{\hbar^{2j-1}} F_j(2q) \delta^{(2j)}(p) + \sum_{j=0}^\infty \frac{1}{2^j j!} \frac{\mu^{2j}\omega^{2j}}{\hbar^{2j}} q^{2j} \int_{\infty}^\infty G_j(v) \exp\qty(-i\frac{p}{\hbar}v) \dd{v} \, .
\end{split}
\end{equation}
The consequences of Hermiticity and time-reversal symmetry are similar to that of the free particle case. Imposing Hermiticity makes the second and third term purely real. For $g(v) = g^*(-v)$, we get $G_s(v) = G^*_s(-v)$ for $s \geq 0$, and so the Fourier transform in the last term is purely real as well. Imposing time-reversal symmetry gives imaginary components in the solution. Imposing both Hermiticity and time-reversal symmetry leaves us with the first and last terms (removing all Dirac delta terms and their derivatives), both purely real.

We therefore see that the general solution of the modified time kernel equation has terms in its Wigner-Weyl transform containing Dirac deltas and their derivatives. Imposing conditions on $\alpha$ (or $\beta$), $f$ and $g$ determines one specific (equivalence class) solution. These give solutions to the time-energy canonical commutation relation. We get the usual time of arrival solution when $\alpha = \beta$ and $f(u) = g(v) = 0$.

\section{Conclusion} \label{sec:conclusion}
The time kernel equation and the accompanying boundary conditions provide a solution to the time-energy canonical commutation relation in position representation, one specific case of which is the time of arrival solution. In this paper, we accomplished the following: (1) we rewrote the boundary conditions in a more convenient form, (2) gave conditions for Hermiticity and time-reversal symmetry, (3) provided some interesting examples of other Hamiltonian conjugate solutions, the shifted time of arrival solution, and (4) considered a modified form of the time kernel equation and studied its solutions. The modified time kernel equation provides an even more general way of constructing Hermitian Hamiltonian conjugates, removing the assumption on the form of the kernel and considering locally integrable functions, which resulted with Dirac deltas in the Wigner-Weyl transform that aren't present in the original kernel solution.

Both of the above methods of solving a differential equation constitute the supraquantization of an operator that is canonically conjugate to the Hamiltonian in position representation. We see that with this requirement, for a particle in 1-dimension under a continuous potential $V(q)$, both solutions give the time of arrival plus some other terms. This time of arrival term is always present, and could be thought of as the ``master time", wherein other possible Hamiltonian conjugates are just shifted time of arrivals.

The (original) time kernel equation solutions can be written as Bender-Dunne operators in position representation \cite{galapon2008}, and are thus Hilbert space operators \cite{bunao2014}. This connection with the Bender-Dunne operators is not present in the modified kernel solution. In position representation, the Bender-Dunne minimal solution gives a kernel that is an analytic function multiplied by a signum function; this does not encompass all possible solutions that the modified time kernel equation provides. The status of the modified time kernel equation solutions as Hilbert space operators is still an unanswered question.

Are all these solutions time observables? While, at the minimum, we require conjugacy with the Hamiltonian and Hermiticity, one can only interpret a particular solution as some time observable by looking at its classical limit. What about its self-adjointness? That would require further study, for instance, of the deficiency indices of these operators \cite{reed1975}. What is the role of time-reversal symmetry? For now, it is unclear whether this is a strict requirement for all time observables.

\section*{Data Availability Statement}
There is no data associated with this manuscript.

\begin{appendices}


\section{The General Potential Solution} \label{app:genpolsoln}
The general potential \eqref{eq:genpol} for the time kernel equation \eqref{eq:tke} gives the recurrence relation \eqref{eq:genpolrecur} for the coefficients of the analytic solution. The boundary conditions \eqref{eq:tkehnbca} give two branches of nonvanishing terms: the first condition gives the time of arrival, and the second condition gives the shifting term. This second condition gives a nonvanishing $\alpha_{0,2N-1}$ for positive integer $N$, and so the recurrence relation \eqref{eq:genpolrecur} makes all $\alpha_{m,n}$ terms with even $n$ vanish.

We solve this following analogous steps as that in \cite{galapon2004}. We now look at the branch of coefficients $\alpha_{m,n}$ with odd $n$. First, we have, for $m \geq 0$,
\begin{equation} \label{eq:am2n-1}
    \alpha_{m,2N-1} = -i^{2N-1} \frac{\beta}{\mu^{2N-2}} \delta_{m,0} \, .
\end{equation}
Using the recurrence relation \eqref{eq:genpolrecur}, we get $\alpha_{m,2N+1}$ for $m \geq 1$,
\begin{equation}
    \alpha_{m,2N+1} = -i^{2N-1} \frac{\beta}{\mu^{2N-2}} \frac{\mu}{2\hbar^2} \frac{1}{2N+1} \frac{a_m}{2^{m-1}} \, .
\end{equation}
Similarly,
\begin{equation}
\begin{split}
    \alpha_{m,2N+3} &= -i^{2N-1} \frac{\beta}{\mu^{2N-2}} \qty(\frac{\mu}{2\hbar^2})^2 \frac{1}{m(2N+1)(2N+3)} \sum_{s=1}^{m-1} \frac{s a_s a_{m-s}}{2^{m-2}} \\
    &\qquad -i^{2N-1} \frac{\beta}{\mu^{2N-2}} \frac{\mu}{2\hbar^2} \frac{1}{m(2N+3)} \frac{a_{m+2}}{2^{m+1}} \binom{m+2}{3} \, .
\end{split}
\end{equation}
This suggests that, for $m \geq 1$ and $j \geq 1$,
\begin{equation} \label{eq:amnasmj}
    \alpha_{m,2N-1+2j} = \sum_{s=0}^{j-1} \qty(\frac{\mu}{2\hbar^2})^{j-s} \alpha_{m,j}^{(s)} \, .
\end{equation}
The recurrence relation \eqref{eq:genpolrecur} gives
\begin{equation}
    \alpha_{m,2N-1+2j} = \frac{\mu}{2\hbar^2} \frac{1}{m(2N-1+2j)} \sum_{s=1}^\infty \frac{a_s}{2^{s-1}} \sum_{k=0}^{[s]} \binom{s}{2k+1} \alpha_{m-s+2k,2N+2j-2k-3} \, .
\end{equation}
Note that $\alpha_{m-s+2k,2N+2j-2k-3}$ is nonvanishing for $m-s+2k \geq 0$ and $2N+2j-2k-3 \geq 2N - 1$. Also note that the binomial coefficient is nonvanishing for $s \geq 2k+1$. We can then rewrite the above equation into
\begin{equation}
    \alpha_{m,2N-1+2j} = \frac{\mu}{2\hbar^2} \frac{1}{m(2N-1+2j)} \sum_{k=0}^{j-1} \sum_{s=2k+1}^{m+2k} \frac{a_s}{2^{s-1}} \binom{s}{2k+1} \alpha_{m-s+2k,2N-1+2(j-k-1)} \, .
\end{equation}
Substituting \eqref{eq:amnasmj} to the right hand side gives
\begin{equation}
\begin{split}
    \alpha_{m,2N-1+2j} &= \frac{\mu}{2\hbar^2} \frac{1}{m(2N-1+2j)} \sum_{k=0}^{j-1} \sum_{s=2k+1}^{m+2k} \frac{a_s}{2^{s-1}} \binom{s}{2k+1} \\
    &\quad \times \sum_{r=0}^{j-k-2} \qty(\frac{\mu}{2\hbar^2})^{j-k-1-r} \alpha^{(r)}_{m-s+2k,j-k-1} \, .
\end{split}
\end{equation}
Since $0 \leq k \leq j - 1$, then the summation in $r$ is nonvanishing up to $j-k-1$, and so we can rewrite this into
\begin{equation}
\begin{split}
    \alpha_{m,2N-1+2j} &= \frac{1}{m(2N-1+2j)} \\
    &\quad \times \sum_{k=0}^{j-1} \sum_{r=0}^{j-k-1} \sum_{s=2k+1}^{m+2k} \frac{a_s}{2^{s-1}} \binom{s}{2k+1} \qty(\frac{\mu}{2\hbar^2})^{j-k-r} \alpha^{(r)}_{m-s+2k,j-k-1} \, .
\end{split}
\end{equation}
Interchanging the $r$ and $k$ summations, we get
\begin{equation}
\begin{split}
    \alpha_{m,2N-1+2j} &= \frac{1}{m(2N-1+2j)} \\
    &\qquad \times \sum_{r=0}^{j-1} \sum_{k=0}^{j-r-1} \sum_{s=2k+1}^{m+2k} \frac{a_s}{2^{s-1}} \binom{s}{2k+1} \qty(\frac{\mu}{2\hbar^2})^{j-k-r} \alpha^{(r)}_{m-s+2k,j-k-1} \, .
\end{split}
\end{equation}
Rewriting this by letting $k$ run from $0$ to $r$ gives
\begin{equation}
\begin{split}
    \alpha_{m,2N-1+2j} &= \sum_{r=0}^{j-1} \qty(\frac{\mu}{2\hbar^2})^{j-r} \\
    &\qquad \times \frac{1}{m(2N-1+2j)} \sum_{k=0}^{r} \sum_{s=2k+1}^{m+2k} \frac{a_s}{2^{s-1}} \binom{s}{2k+1} \alpha^{(r-k)}_{m-s+2k,j-k-1} \, . 
\end{split}
\end{equation}
Comparing this with \eqref{eq:amnasmj}, we thus get the recurrence relation
\begin{equation} \label{eq:asmjrecur}
    \alpha_{m,j}^{(s)} = \frac{1}{m(2N-1+2j)} \sum_{r=0}^{s} \sum_{n=2r+1}^{m+2r} \frac{a_n}{2^{n-1}} \binom{n}{2r+1} \alpha^{(s-r)}_{m-n+2r,j-r-1} \, .
\end{equation}
From \eqref{eq:amnasmj}, we see that the nonvanishing coefficients contribute to $v^{2N-1+2j}$, giving a Wigner-Weyl contribution \eqref{eq:wignerweyl} proportional to
\begin{equation}
    \mathcal{T}_\hbar \propto \frac{1}{\hbar} \alpha_{m,2N-1+2j} \int_{-\infty}^\infty v^{2N-1+2j} \sgn(v) \exp\qty(-i\frac{p}{\hbar}v) \dd{v} \, ,
\end{equation}
which gives, upon using \eqref{eq:fouriertnsgnt},
\begin{equation}
    \mathcal{T}_\hbar \propto \sum_{s=0}^{j-1} (-1)^j \qty(\frac{\mu}{2})^{j-s} (2N-1+2j)! \alpha^{(s)}_{m,j} \frac{\hbar^{2N-1+2s}}{p^{2j}} \, .
\end{equation}
From our results for linear systems, we infer that the Wigner-Weyl transform here should look like
\begin{equation}
    \mathcal{T}_\hbar = \beta \frac{(2N-1)!}{2^{N-1}} \frac{\hbar^{2N-1}}{\mu^{3(N-1)}} \frac{1}{H^N} \, .
\end{equation}
This means that we are only interested at the $\hbar^{2N-1}$ term in the Wigner-Weyl transform since that is the only contributing factor to our desired classical limit. This term is the $s = 0$ term in our calculations. The leading $\hbar$ correction is then $\mathcal{O}(\hbar^{2N+1})$, corresponding to $s=1$. 

Since only the $s=0$ terms correspond to the classical limit, then from \eqref{eq:asmjrecur}, the recurrence relation that we are interested in studying is
\begin{equation} \label{eq:a0mjrecur}
    \alpha^{(0)}_{m,j} = \frac{1}{m(2N-1+2j)} \sum_{n=1}^m \frac{n a_n}{2^{n-1}} \alpha^{(0)}_{m-n,j-1} \, ,
\end{equation}
where we have let $\alpha^{(s)}_{m,j}$ for $s > 0$ vanish.

We note in passing that the leading $\hbar$ correction of the time of arrival result is $\mathcal{O}(\hbar^2)$ for nonlinear systems \cite{galapon2004}. We then need to let $\hbar$ vanish if we are to recover the correct classical limit for the time of arrival. One could choose $\beta$ such that $\beta\hbar^{2N-1}$ does not vanish, i.e., $\beta \propto 1/\hbar^{2N-1}$, so that the leading $\hbar$ correction for $T_C$ becomes $\mathcal{O}(\hbar^2)$. In this scenario, letting $\mathcal{O}(\hbar^2)$ vanish won't remove the $H^{-N}$ term.

Going back, since $\alpha_{m,2N-1}$ is given by \eqref{eq:am2n-1}, then
\begin{equation} \label{eq:a0m0}
    \alpha^{(0)}_{m,0} = -i^{2N-1} \frac{\beta}{\mu^{2N-2}} \delta_{m,0} \, ,
\end{equation}
\begin{equation}
    \alpha^{(0)}_{m,1} = -i^{2N-1} \frac{\beta}{\mu^{2N-2}} \frac{1}{2N+1} \frac{a_m}{2^{m-1}} \, ,
\end{equation}
\begin{equation}
    \alpha^{(0)}_{m,2} = -i^{2N-1} \frac{\beta}{\mu^{2N-2}} \frac{1}{m(2N+3)(2N+1)} \frac{1}{2^{m-2}} \sum_{s=1}^m s a_s a_{m-s} \, ,
\end{equation}
\begin{equation}
    \alpha^{(0)}_{m,3} = -i^{2N-1} \frac{\beta}{\mu^{2N-2}} \frac{1}{m(2N+5)(2N+3)(2N+1)} \frac{1}{2^{m-3}} \sum_{s=1}^m \frac{s a_s}{m-s} \sum_{r=1}^{m-s} r a_r a_{m-s-r} \, ,
\end{equation}
where we have used the recurrence relation \eqref{eq:a0mjrecur} to get the other nonvanishing terms. We infer that
\begin{equation} \label{eq:a0mjcmj1}
    \alpha^{(0)}_{m,j} = -i^{2N-1} \frac{\beta}{\mu^{2N-2}} \frac{\Gamma\qty(N+\frac{1}{2})}{\Gamma\qty(N+\frac{1}{2}+j)} \frac{1}{2^m} C_{m,j} \, .
\end{equation}
Substituting this to the right side of \eqref{eq:a0mjrecur} gives
\begin{equation} \label{eq:a0mjcmj2}
    \alpha^{(0)}_{m,j} = -i^{2N-1} \frac{\beta}{\mu^{2N-2}} \frac{\Gamma\qty(N+\frac{1}{2})}{\Gamma\qty(N+\frac{1}{2}+j)} \frac{1}{2^m} \frac{1}{m} \sum_{s=1}^m s a_s C_{m-s,j-1} \, .
\end{equation}
Comparing \eqref{eq:a0mjcmj1} and \eqref{eq:a0mjcmj2}, we get a recurrence relation for $C_{m,j}$,
\begin{equation} \label{eq:cmjrecur}
    C_{m,j} = \frac{1}{m} \sum_{s=1}^m s a_s C_{m-s,j-1} \, .
\end{equation}
From \eqref{eq:a0m0} and \eqref{eq:a0mjcmj1}, we see that $\alpha^{(0)}_{m,0} = -i \frac{\beta}{2^{2N-2}} \delta_{m,0} = -i \frac{\beta}{2^{2N-2}} \frac{1}{2^m} C_{m,0}$, so $C_{m,0} = 2^m \delta_{m,0}$, i.e.,
\begin{equation} \label{eq:cm0}
    C_{m,0} = \delta_{m,0} \, .
\end{equation}

We now go back to the time kernel solution, which takes the form
\begin{equation}
    T(u,v) = T_{\text{TOA}}(u,v) + \sum_{j=0}^\infty \sum_{m=0}^\infty \alpha_{m,2N-1+2j} u^m v^{2N-1+2j} \, .
\end{equation}
From \eqref{eq:amnasmj}, this becomes
\begin{equation}
    T(u,v) = T_{\text{TOA}}^{(0)}(u,v) + \sum_{j=0}^\infty \sum_{m=0}^\infty \qty(\frac{\mu}{2\hbar^2})^j \alpha^{(0)}_{m,j} u^m v^{2N-1+2j} \, ,
\end{equation}
wherein we only consider the $s = 0$ term. In the time of arrival solution, only the $s = 0$ term is taken as well \cite{galapon2004}. With our condition that $\beta\hbar^{2N-1}$ does not vanish, we see that this equation is the leading order solution to the general potential case, as we have let $\order{\hbar^2}$ vanish. To continue, we use \eqref{eq:a0mjcmj1} to get
\begin{equation}
    T(u,v) = T_{\text{TOA}}^{(0)}(u,v) - i^{2N-1} \frac{\beta}{\mu^{2N-2}} \sum_{j=0}^\infty \sum_{m=0}^\infty \qty(\frac{\mu}{2\hbar^2})^j \frac{\Gamma\qty(N+\frac{1}{2})}{\Gamma\qty(N+\frac{1}{2}+j)} \frac{1}{2^m} C_{m,j} u^m v^{2N-1+2j} \, .
\end{equation}
Thus, the Wigner-Weyl transform \eqref{eq:wignerweyl} gives
\begin{equation}
\begin{split}
    \mathcal{T}_\hbar(q,p) &= t_{\text{TOA}}(q,p) - \frac{\mu}{\hbar} \frac{\beta}{2^{2N-2}} \sum_{j=0}^\infty \sum_{m=0}^\infty \qty(\frac{\mu}{2\hbar^2})^j \frac{\Gamma\qty(N+\frac{1}{2})}{\Gamma\qty(N+\frac{1}{2}+j)} \frac{1}{2^m} C_{m,j} (2q)^m \\
    &\qquad \qquad \qquad \quad \times \int_{-\infty}^\infty v^{2N-1+2j} \sgn(v) \exp\qty(-i\frac{p}{\hbar}v) \dd{v} \, ,
\end{split}
\end{equation}
which becomes, after using \eqref{eq:fouriertnsgnt} and some simplifications,
\begin{equation}
\begin{split}
    \mathcal{T}_\hbar(q,p) &= t_{\text{TOA}}(q,p) + \beta \frac{(2N-1)!}{2^{N-1}} \frac{\hbar^{2N-1}}{\mu^{3(N-1)}} \\
    &\qquad \qquad \qquad \quad \times \qty(\frac{2\mu}{p^2})^N \sum_{k=0}^\infty (-1)^k \frac{\Gamma(N+k)}{k!\Gamma(N)} \qty(\frac{2\mu}{p^2})^k k! \sum_{m=0}^\infty C_{m,k} q^m \, .
\end{split}
\end{equation}
Note that for the general potential \eqref{eq:genpol},
\begin{equation}
    \frac{1}{H^N} = \qty(\frac{2\mu}{p^2})^N \sum_{k=0}^\infty (-1)^k \frac{\Gamma(N+k)}{k!\Gamma(N)} \qty(\frac{2\mu}{p^2})^k \qty(\sum_{s=1}^\infty a_s q^s)^k \, ,
\end{equation}
for sufficiently small $q$. We then are left with showing that
\begin{equation} \label{eq:cmk}
    k! \sum_{m=0}^\infty C_{m,k} q^m = \qty(\sum_{s=1}^\infty a_s q^s)^k \, ,
\end{equation}
so that the shifting term approaches the correct classical limit.

Firstly, for $k=0$,
\begin{equation}
    \sum_{m=0}^\infty C_{m,0} q^m = 1 \, ,
\end{equation}
and so, by \eqref{eq:cm0}, we know that this equality holds. To show that this holds for $k > 1$, we first differentiate both sides of \eqref{eq:cmk} with respect to $q$,
\begin{equation}
    k! \sum_{m=0}^\infty m C_{m,k} q^{m-1} = k \qty(\sum_{s=1}^\infty a_s q^s)^{k-1} \sum_{r=1}^\infty r a_r q^{r-1} \, .
\end{equation}
Suppose \eqref{eq:cmk} is true; we can rewrite the above equation into
\begin{equation}
    k! \sum_{m=0}^\infty m C_{m,k} q^{m-1} = k \qty((k-1)! \sum_{n=0}^\infty C_{n,k-1} q^n) \sum_{r=1}^\infty r a_r q^{r-1} \, .
\end{equation}
We rewrite the double sum as
\begin{equation}
    \sum_{m=0}^\infty m C_{m,k} q^{m-1} = \sum_{m=0}^\infty \sum_{j=0}^m (m-j) a_{m-j} C_{j,k-1} q^{m-1} \, ,
\end{equation}
or,
\begin{equation}
    \sum_{m=0}^\infty m C_{m,k} q^{m-1} = \sum_{m=0}^\infty \sum_{s=0}^m s a_s C_{m-s,k-1} q^{m-1} \, ,
\end{equation}
and thus, we obtain
\begin{equation}
    m C_{m,k} = \sum_{s=0}^m s a_s C_{m-s,k-1} \, .
\end{equation}
By \eqref{eq:cmjrecur}, we know that this equality holds as well, implying that \eqref{eq:cmk} is indeed true.

\end{appendices}



\begin{thebibliography}{99}
\bibitem{gotay2000}
M. J. Gotay, “Obstructions to quantization,” in \emph{Mechanics: from theory to computation} (Springer, 2000), pp. 171–216.

\bibitem{muga2008}
G. Muga, R. S. Mayato, and I. Egusquiza, eds., \emph{Time in quantum mechanics - vol. 1}, 2nd edition, Vol. 734, Lecture Notes in Physics (Springer, Berlin Heidelberg, 2008).

\bibitem{muga2009}
G. Muga, A. Ruschhaupt, and A. del Campo, eds., \emph{Time in quantum mechanics - vol. 2}, Vol. 789, Lecture Notes in Physics (Springer, Berlin Heidelberg, 2009).

\bibitem{pauli1980}
W. Pauli, \emph{General principles of quantum mechanics} (Springer-Verlag Berlin Heidelberg, 1980).

\bibitem{galapon2002}
E.  Galapon,  “Pauli’s  theorem  and  quantum  canonical  pairs:  the  consistency  of  a  bounded,  self–adjoint time operator canonically conjugate to a Hamiltonian with non–empty point spectrum,” in \emph{Proceedings of the Royal Society of London A - Mathematical, Physical \& Engineering Sciences}, Vol. 458 (The Royal Society, 2002), pp. 451–472.

\bibitem{aharonov1961}
Y.  Aharonov  and  D.  Bohm,  “Time  in  the  quantum  theory  and  the  uncertainty  relation  for  time  and energy,” \emph{Physical Review} \textbf{122}, 1649 (1961).

\bibitem{galapon2018}
E. A. Galapon and J. J. P. Magadan, “Quantizations of the classical time of arrival and their dynamics,” \emph{Annals of Physics} \textbf{397}, 278–302 (2018).

\bibitem{kijowski1974}
J.  Kijowski,  “On  the  time  operator  in  quantum  mechanics  and  the  Heisenberg  uncertainty  relation  for energy and time,” \emph{Reports on Mathematical Physics} \textbf{6}, 361–386 (1974).

\bibitem{bauer1983}
M. Bauer, “A time operator in quantum mechanics,” \emph{Annals of Physics} \textbf{150}, 1–21 (1983).

\bibitem{razavy1969}
M.  Razavy,  “Quantum-mechanical  conjugate  of  the  hamiltonian  operator,”  \emph{Il  Nuovo  Cimento  B (1965-1970)} \textbf{63}, 271–308 (1969).

\bibitem{olkhovsky1974}
V. Olkhovsky, E. Recami, and A. Gerasimchuk, “Time operator in quantum mechanics,” \emph{Il Nuovo Cimento A (1965-1970)} \textbf{22}, 263–278 (1974).

\bibitem{goto1981}
T. Goto, K. Yamaguchi, and N. Sudo, “On the time operator in quantum mechanics: three typical examples,” \emph{Progress of Theoretical Physics} \textbf{66}, 1525–1538 (1981).

\bibitem{bender1989exact}
C. M. Bender and G. V. Dunne, “Exact solutions to operator differential equations,” \emph{Physical Review D} \textbf{40}, 2739 (1989).

\bibitem{bender1989integration}
C. M. Bender and G. V. Dunne, “Integration of operator differential equations” \emph{Physical Review D} \textbf{40}, 3504 (1989).

\bibitem{galapon2004}
E. A. Galapon, “Shouldn’t there be an antithesis to quantization?,” \emph{Journal of Mathematical Physics} \textbf{45}, 3180 (2004).

\bibitem{galapon2008}
E. A. Galapon and A. Villanueva, “Quantum first time-of-arrival operators,” \emph{Journal of Physics A: Mathematical and Theoretical} \textbf{41}, 455302 (2008).

\bibitem{bunao2014}
J. Bunao and E. A. Galapon, “The Bender-Dunne basis operators as Hilbert space operators,” \emph{Journal of Mathematical Physics} \textbf{55}, 022102 (2014).

\bibitem{domingo2015}
H. B. Domingo and E. A. Galapon, “Generalized Weyl transform for operator ordering: polynomial functions in phase space,” \emph{Journal of Mathematical Physics} \textbf{56}, 022104 (2015).

\bibitem{bagunu2021}
R. J. C. Bagunu and E. A. Galapon, “Solutions to the time-energy canonical commutation relation using Weyl, symmetric, and Born-Jordan basis operators,” in \emph{Proceedings of the Samahang Pisika ng Pilipinas}, Vol. 39 (2021), SPP-2021-PC–04.

\bibitem{mackey1968}
G. W. Mackey, \emph{Induced representations of groups and quantum mechanics} (Benjamin, 1968).

\bibitem{bohm1978}
A. Bohm, \emph{Rigged Hilbert space and quantum mechanics}, Vol. 78, Lecture Notes in Physics (Springer-Verlag Berlin Heidelberg, 1978).

\bibitem{domingo2004}
H. B. Domingo, \emph{Time of arrival quantum-classical correspondence in rigged Hilbert space}, Master’s Thesis, University of the Philippines Diliman, 2004.

\bibitem{freiling2008}
G. Freiling and V. Yurko, \emph{Lectures on differential equations on mathematical physics: a first course} (Nova Science Publishers, Inc, 2008).

\bibitem{caballar2009}
R. C. F. Caballar and E. A. Galapon, “Characterizing  multiple  solutions  to  the  time–energy  canonical commutation relation via quantum dynamics,” \emph{Physics Letters A} \textbf{373}, 2660–2666 (2009).

\bibitem{caballar2010}
R. C. F. Caballar, L. R. Ocampo, and E. A. Galapon, “Characterizing multiple solutions to the time-energy canonical commutation relation via internal symmetries,” \emph{Physical Review A} \textbf{81}, 062105 (2010).

\bibitem{magadan2018}
J. J. Magadan and E. Galapon, “Solutions to the time-energy canonical commutation relation for harmonicoscillator potential," in \emph{Proceedings of the 36th Samahang Pisika ng Pilipinas} physics conference (2018).

\bibitem{zemanian1987}
A. H. Zemanian, \emph{Distribution theory and transform analysis: an introduction to generalized functions, with applications} (Dover Publications, Inc, 1987).

\bibitem{reed1975}
M. Reed and B. Simon, \emph{Methods of modern and mathematical physics II: Fourier analysis, self-adjointness}, Vol. 2 (Elsevier, 1975).

\end{thebibliography}

\end{document}